\documentclass[11pt]{article}

\usepackage{a4wide}

\usepackage{amssymb}
\usepackage{amsmath}

\usepackage[latin1]{inputenc}

\usepackage{graphicx}

\usepackage{stmaryrd}

\usepackage{tikz}
\usetikzlibrary{arrows,shapes.gates.logic.US,shapes.gates.logic.IEC,calc} 

\usepackage{xy}
\xyoption{import}

\usepackage{hyperref}

\newcommand{\sder}[1]               {{\begin{center}
                                      \begin{tabular}{llr}
                                        #1
                                      \end{tabular}
                                      \end{center}}}

\newcommand{\culprit}                   {\textrm{culprit}^{J}_{C}}
\newcommand{\inplace}                   {\textrm{inplace}^{J}_{C}}
\newcommand{\capable}                   {\textrm{capable}^{J}_{C}}
\newcommand{\motive}                   {\textrm{motive}^{J}_{C}}
\newcommand{\poor}                   {\textrm{poor}^{J}}
\newcommand{\friend}                   {\textrm{friend}^{J}}
\newcommand{\guilty}                   {\textrm{guilty}^{J}_{C'}}
\newcommand{\similar}                   {\textrm{similar}^{C}_{C'}}
\newcommand{\divergence}              {\textrm{divergence}^{J}_{C}}

\newcommand{\Jculprit}                   {t\cdot\textrm{culprit}^{J}_{C}}
\newcommand{\Jinplace}                   {t\cdot\textrm{inplace}^{J}_{C}}
\newcommand{\Jcapable}                   {t\cdot\textrm{capable}^{J}_{C}}
\newcommand{\Jmotive}                   {t\cdot\textrm{motive}^{J}_{C}}
\newcommand{\Jpoor}                   {t\cdot\textrm{poor}^{J}}
\newcommand{\Jfriend}                   {t\cdot\textrm{friend}^{J}}
\newcommand{\Jguilty}                   {t'\cdot\textrm{guilty}^{J}_{C'}}
\newcommand{\Csimilar}                   {t\cdot\textrm{similar}^{C}_{C'}}
\newcommand{\Jdivergence}              {t\cdot\textrm{divergence}^{J}_{C}}

\newcommand{\Cut}                   {\textrm{Cut}}
\newcommand{\rAT}                   {\textrm{AT}}
\newcommand{\rAI}                   {\textrm{AR}}

\newcommand{\rAC}                   {\textrm{AC}}

\newcommand{\rBCL}                   {\textrm{SCL}}
\newcommand{\rBCR}                   {\textrm{SCR}}

\newcommand{\rBKR}                   {\textrm{SKR}}
\newcommand{\rBKL}                   {\textrm{SKL}}

\newcommand{\rKL}                   {\textrm{KL}}
\newcommand{\rBP}                   {\textrm{SP}}
\newcommand{\rKP}                   {\textrm{KP}}

\newcommand{\rNB}                   {\textrm{NS}}

\newcommand{\rNK}                   {\textrm{NK}}
\newcommand{\CR}                   {\textrm{CR}}

\newcommand{\ER}                   {\textrm{ER}}
\newcommand{\EC}                   {\textrm{EC}}
\newcommand{\ET}                   {\textrm{ET}}
\newcommand{\ES}                   {\textrm{ES}}
\newcommand{\rAML}                   {\textrm{AML}}
\newcommand{\rAMR}                   {\textrm{AMR}}

\newcommand{\var}                   {{\textrm{var}}}

\newcommand{\lap}                  {\mathbin{\trianglelefteq}}
\newcommand{\lapdiv}                  {\mathbin{\trianglelefteq}_{\textrm{divergence}^{\textrm{J}}_{\textrm{C}}}}
\newcommand{\lapmot}                  {\mathbin{\trianglelefteq}_{\textrm{motive}^{\textrm{J}}_{\textrm{C}}}}

\newcommand{\Ax}                    {\textrm{Ax}}

\newcommand{\qurtains}             {\hfill{QED}}

\ifx \theorem \undefined
\newtheorem{definition}{\vspace{1mm}Definition}[section]

\newtheorem{example}[definition]{\vspace{1mm}Example}

\newtheorem{prop}[definition]{\vspace{1mm}Proposition}

\newenvironment{proof}{\begin{trivlist}\item{\bf Proof:}}{\qurtains\end{trivlist}}

\newcommand{\lnecb}                  {\mathop{\boxdot}}

\newcommand{\sat}                    {\Vdash}
\newcommand{\ent}                    {\vDash}
\newcommand{\der}                    {\vdash}

 \newcommand{\FE}{\mathrm{SC_1}}
\newcommand{\Sans}{\mathrm{SC_2}}
\newcommand{\DHS}{\mathrm{SD}}
\newcommand{\R}{\mathrm{SA_1}}
\newcommand{\ROne}{\mathrm{SA_2}}
\newcommand{\Russia}{\mathrm{A}}

\newcommand{\Sandworm}{\mathrm{HG}}
\newcommand{\BBC}{\mathrm{Ag}}
\newcommand{\U}{\mathrm{NV}}
\newcommand{\In}{\mathrm{F_1}}
\newcommand{\SansIn}{\mathrm{F_2}}
\newcommand{\SC}{\mathrm{SC_{3}}}
\newcommand{\Blog}{\mathrm{Bl}}

\newcommand{\geoSourceIP}{\textrm{geoSourceIP}}
\newcommand{\Att}{\textrm{PGA}}
\newcommand{\usedMalware}{\textrm{usedMalware}}
\newcommand{\BEtwo}{\textrm{BE2}}
\newcommand{\BEthree}{\textrm{BE3}}
\newcommand{\possibleCulprit}{\textrm{possibleCulprit}}
\newcommand{\KillDisk}{\textrm{KillDisk}}
\newcommand{\simil}{\textrm{similar}}
\newcommand{\usedBy}{\textrm{usedBy}}
\newcommand{\targetVictim}{\textrm{targetVictim}}
\newcommand{\Opponent}{\textrm{Opponent}}
\newcommand{\econLoss}{\textrm{econLoss}}
\newcommand{\repLoss}{\textrm{repLoss}}
\newcommand{\limitedDamage}{\textrm{limitedDamage}}
\newcommand{\found}{\textrm{found}}
\newcommand{\spearPhishing}{\textrm{spearPhishing}}
\newcommand{\infection}{\textrm{infection}}
\newcommand{\backdoor}{\textrm{backdoor}}
\newcommand{\exploitedVuln}{\textrm{exploitedVuln}}
\newcommand{\econMotives}{\textrm{econMotives}}
\newcommand{\attacker}{\textrm{econMotives}}
\newcommand{\victimAccess}{\textrm{victimAccess}}
\newcommand{\Victim}{\textrm{U}}
\newcommand{\phishEmails}{\textrm{phishEmails}}
\newcommand{\HMIvuln}{\textrm{HMIvuln}}
\newcommand{\news}{\textrm{news}}
\newcommand{\Blackout}{\textrm{Blackout}}
\newcommand{\polConflict}{\textrm{polConflict}}
\newcommand{\started}{\textrm{started}}

\newcommand{\logic}{Time-Stamped Claim Logic}

\def\aiOrcid{\includegraphics[scale=0.5]{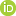}}
\newcommand{\orcidID}[1]{\href{https://orcid.org/#1}{\textcolor[HTML]{A6CE39}{\aiOrcid}}}

\begin{document}

\title{Time-Stamped Claim Logic} 

\author{João Rasga$^1$ \orcidID{0000-0002-1239-8496}, 
Cristina Sernadas$^1$ \orcidID{0000-0002-5510-3512},
Erisa Karafili$^2$ \orcidID{0000-0002-8250-4389}, 
Luca Vigan\`o$^3$ \orcidID{0000-0001-9916-271X}
\\[1.5mm]
$^1${\scriptsize Departamento de Matemática, Instituto Superior Técnico and CMAFcIO, ULisboa, Portugal}\\[0mm] 
{\scriptsize \tt \{joao.rasga,cristina.sernadas\}@tecnico.ulisboa.pt}\\[0mm]
$^2${\scriptsize Department of Computing, Imperial College London, UK}\\[0mm]
{\scriptsize \tt  e.karafili@imperial.ac.uk}\\[0mm]
$^3${\scriptsize Department of Informatics, King's College London, UK}\\[0mm]
{\scriptsize \tt  luca.vigano@kcl.ac.uk}}

\date{ } 

\maketitle

\begin{abstract} 

The main objective of this paper is to define a logic for reasoning about distributed time-stamped claims. Such a logic is interesting for theoretical reasons, i.e., as a logic \emph{per se}, but also because it has a number of practical applications, in particular when one needs to reason about a huge amount of pieces of evidence collected from different sources, where some of the pieces of evidence may be contradictory and some sources are considered to be more trustworthy than others. We introduce the Time-Stamped Claim Logic including a sound and complete sequent calculus that allows one to reduce the size of the collected set of evidence and removes inconsistencies, i.e., the logic ensures that the result is consistent with respect to the trust relations considered.  In order to show how Time-Stamped Claim Logic can be used in practice, we consider a concrete cyber-attribution case study.
\\[2mm]
\end{abstract}

\section{Introduction}\label{sec:introduction} 

The main objective of this paper is to define a logic for reasoning about distributed time-stamped claims. Such a logic is interesting for theoretical reasons, i.e., as a logic \emph{per se}, but also because it has a number of practical applications, most notably the ability to reason about the attribution of cyber-attacks. 

When reasoning about a cyber-attack, a digital forensics analyst typically collects a huge amount of pieces of evidence from different sources. Some of the  pieces of evidence may be contradictory and the analyst might consider some sources to be more trustworthy than others. Inferring conclusions from such evidence thus requires particular care and time. A similar problem is faced by historians when they are trying to date particular historic events.

The Time-Stamped Claim Logic that we introduce in this paper provides the forensics analyst with a sound and complete means to reduce the size of the collected set of evidence and remove inconsistencies, i.e., the logic ensures that the result is consistent with respect to the trust relations that the analyst considers to hold.
More specifically, the Time-Stamped Claim Logic is a monotonic propositional logic whose language  contains assertions of different kinds: labeled assertions to expressing statements of agents about time-stamped claims and relational assertions indexed by propositional symbols that are used to relate agents. In particular, a trust relation with respect to each propositional subject is defined between agents.

We formalize a Gentzen calculus for our logic, which allows one to infer a time-stamped claim whenever (i) there is an agent that states the claim and (ii) each agent that is more trustworthy with respect to the subject does not state the opposite claim. 
We define a modal and many-valued semantics for our logic, and prove that the calculus is sound and complete with respect to this semantics. 
In order to show how Time-Stamped Claim Logic can be used in practice, we consider, as a proof of concept, a concrete cyber-attribution case study, inspired by the Ukraine Power Grid Attack that occurred in December 2015.

We proceed as follows. In Section~\ref{sec:language}, we define the language of Time-Stamped Claim Logic. In Section~\ref{sec:calculus}, we formalize a Gentzen calculus for the logic together with a running example. In Section~\ref{sec:semantics}, 
we introduce the semantics, and we then prove the soundness and the completeness of the calculus in Section~\ref{sec:soundness} and Section~\ref{sec:completeness}, respectively.
In Section~\ref{sec:case-study}, we apply Time-Stamped Claim Logic to a realistic case study taken from the cyber-security area. We discuss the most relevant related work in Section~\ref{sec:Comparison} and provide some concluding remarks and ideas for future work in Section~\ref{sec:concrems}.

\section{The Language of Time-Stamped Claim Logic}\label{sec:language}

In this section, we introduce the language of the \logic{}.
\begin{definition}
Let $P$ be a non-empty set of propositional symbols and assume fixed non-empty pairwise disjoint sets $X_A$ and $X_T$ of variables, which represent agents names and time points, respectively. 
The set $K_P$ of \emph{time-stamped propositional claims} is defined as
$$
K_P = \{{-}(t \cdot p),\ t \cdot p \mid t \in X_T \text{ and } p \in P\}
$$
\hfill $\blacksquare$
\end{definition}

\begin{definition}
The \emph{language} $L_P$ of \emph{assertions} is defined as follows:
\begin{itemize}
\item $K_P \subseteq  L_P$;
\item $t_1 \cong t_2 \in L_P$ whenever $t_1,t_2 \in X_T$;
\item $a_1 \lap_p a_2 \in L_P$ whenever $a_1,a_2 \in X_A$ and $p \in X_P$;
\item $\forall x. \, x \lap_p a$ whenever $a \in X_A$ and $p \in P$;
\item $a: \phi \in L_P$ whenever $a \in X_A$ and $\phi \in K_P$;
\item $a: \lnecb \phi \in L_P$ whenever $a \in X_A$ and $\phi \in K_P$;
\item $a :: (\phi_1,\, \dots, \, \phi_n / \phi)$ whenever $a \in X_A$ 
 and $\phi,\phi_1,\dots,\phi_n \in K_P$. \hfill $\blacksquare$
\end{itemize}
\end{definition}

Let us briefly discuss the intuitive meaning of these assertions:
\begin{itemize}
\item $t_1 \cong t_2 \in L_P$ establishes that $t_1$ and $t_2$ are equivalent time points.
\item $a_1 \lap_p a_2$ establishes that agent $a_2$ is more trustworthy with respect to $p$ than agent $a_1$.
\item $\forall x. \, x \lap_p a$ establishes that $a$ is the most trustworthy agent with respect to statements about $p$. 
\item $a: \phi$ establishes that $a$ states $\phi$.

\item $a: \lnecb t\cdot p$ establishes that there are no agents more trustworthy with respect to $p$ than $a$ that state ${-} (t\cdot p)$. 
That is, each agent more trustworthy with respect to $p$ than $a$ does not claim ${-} (t\cdot p)$. Similarly, for
$a: \lnecb {-} (t\cdot p)$.

\item $a ::(\phi_1,\, \dots, \, \phi_n / \phi)$ establishes that $a$ states $\phi$ conditional to statements $\phi_1,\dots, \phi_n$. 
\end{itemize}
Given the relation of $\phi$ with respect to the other statements we will refer to $a ::(\phi_1,\, \dots, \, \phi_n / \phi)$ as a \emph{derived evidence}. We will instead refer to the previous two assertions $a: \phi$ and $a: \lnecb \phi$ simply as \emph{evidence}.

Let us now illustrate how our logic is used by means of a running example. 

\begin{example}\em
\label{example:one}
Let us suppose that we want to model a crime situation C that possibly occurred at time $t$ in which there is a potential culprit J and some sources $a_1$, $a_2$ and $a_3$ providing statements on the case. 
Let us further suppose that for the source $a_1$ a person is a culprit provided that person was at the scene of the crime, is capable of committing the crime and has a motive. 
 This can be expressed by the assertion
$$
a_1::(\Jinplace,\Jcapable, \Jmotive\; /\; \Jculprit)
$$ 
over $P=\{\inplace,\capable, \motive,\culprit\}$.
Moreover, we can write the assertion
$$
\forall x. \, x \lapmot a_3
$$
to express that  $a_3$ is a source recognized to be the most trustworthy with respect to knowing the motive of the potential culprit $J$. 
 Furthermore, the assumptions
\begin{itemize}
\item $a_1$ is less trustworthy than $a_2$ with respect to knowing the motive of $J$ to have committed the crime $C$,
\item $a_1$ claims that $J$ does not have a motive to have committed the crime $C$ at time $t$,
\item $a_2$ claims that $J$ has a motive to have committed the crime $C$ at time $t$
\end{itemize}
can be expressed in the Time-Stamped Claim Logic by means of the assertions
\begin{itemize}
\item $a_1 \lapmot a_2$, 
\item $a_1:{-} (\Jmotive)$ and
\item $a_2:\Jmotive$,
\end{itemize}
respectively. Hence, if $a_3$ does not have an opinion about whether or not $J$ has a motive to have committed the crime $C$ at time $t$,
then 
\begin{itemize}
\item $a_3: \lnecb \, \Jmotive$ and $a_3: \lnecb  {-} (\Jmotive)$ should hold, 
\item $a_2:\lnecb \, \Jmotive$ should hold and $a_2: \lnecb  {-} (\Jmotive)$ should not hold, and 
\item both $a_1: \lnecb \, \Jmotive$ and $a_1: \lnecb  {-} (\Jmotive)$ should not hold.
\hfill $\blacksquare$
\end{itemize}
\end{example}

Before we introduce the calculus of our logic, which will allow us to draw conclusions from assertions like the ones in Example~\ref{example:one}, let us define some useful notation. In the following, we will write 
$$
\var_A: \wp L_P \to \wp X_A
$$ 
to denote the map that assigns to each set of assertions the set of elements of $X_A$ that occur in it; mutatis mutandis for $\var_T: \wp L_P \to \wp X_T$ and $\var_P: \wp L_P \to \wp P$. We may confuse a singleton set with its unique element. Given $P' \subseteq P$, $X'_A \subseteq X_A$ and $X'_T \subseteq X_T$, we write $L_{P'}^{X'_A,X'_T}$ to denote the subset of $L_P$ including all the assertions using only symbols in $P'$ and variables in $X'_A$ and $X'_T$.

\section{The Calculus of  Time-Stamped Claim Logic}
\label{sec:calculus}

Let us begin by recalling what is a sequent and an inference rule. Let $P$ be a non-empty set of propositional symbols. 
A \emph{sequent} over $P$ is a pair $(\Gamma,\Delta)$ and is denoted by
$$
\Gamma \to \Delta\,,
$$ 
where $\Gamma$ and $\Delta$ are finite multisets of formulas in $L_P$. An \emph{inference rule} over $P$ is of the form 
$$
\frac{\Gamma_1 \to \Delta_1 \quad \cdots \quad \Gamma_n \to \Delta_n}{\Gamma \to \Delta}\;  F
$$
where $F$ is a set of fresh variables with at most a variable in $X_A$. 

For convenience, below, we will sometimes write 
$$
{-}{-} (t \cdot p)
$$ 
to mean $t \cdot p$ and \emph{vice-versa}, in order to avoid the replication 
of rules $\rNB$, $\rNK$, $\rBCR$, $\rBCL$, $\rBKR$, and $\rBKL$.

\begin{definition}
The calculus over $P$ is composed by the following inference rules:

\begin{itemize}
\item \emph{Axiom} (\Ax): 
$$
\frac{}{\beta,\Gamma \to \Delta,\beta}
$$
where $\beta$ is of the form $a: \phi$,  $t_1 \cong t_2$ or $a_1 \lap_p a_2$
\item \emph{Cut} (\Cut):
$$
\frac{\beta, \Gamma \to \Delta \qquad \Gamma \to \Delta,\beta}{\Gamma \to \Delta}
$$
where $\beta$ is of the form $a: \phi$, $t_1 \cong t_2$ or $a_1 \lap_p a_2$

\item \emph{Reflexivity of $\cong$} (\ER):
$$
\frac{}{\Gamma \to \Delta, t \cong t}
$$

\item \emph{Symmetry of $\cong$} (\ES):
$$
\frac{}{t_1 \cong t_2,\Gamma \to \Delta,t_2 \cong t_1}
$$

\item \emph{Transitivity of $\cong$} (\ET):
$$
\frac{}{t_1 \cong t_2,t_2 \cong t_3,\Gamma \to \Delta, t_1 \cong t_3}
$$

\item \emph{Congruence of $\cong$} (\EC):
$$
\frac{}{t_1 \cong t_2,[a:\phi]^{t_1}_{t_2},\Gamma \to \Delta,a:\phi}
$$ 
where $[a:\phi]^{t_1}_{t_2}$ is a formula obtained from $a:\phi$ by replacing $t_1$ by $t_2$

\item \emph{Agent preference transitivity} (\rAT): 
$$
\frac{}{a_1 \lap_{p}a_2,a_2 \lap_{p} a_3,\Gamma \to \Delta,  a_1 \lap_{p} a_3}
$$

\item \emph{Agent preference reflexivity} (\rAI): 
$$
\frac{}{\Gamma \to \Delta,a \lap_{p}a}
$$

\item \emph{Agent preference congruence} (\rAC):
$$
\frac{}{a_1 \lap_p a_2,a_2 \lap_p a_1,[a:\phi]^{a_1}_{a_2},\Gamma \to \Delta,a:\phi}
$$
where $[a:\phi]^{a_1}_{a_2}$ is a formula obtained from $a:\phi$ by replacing $a_1$ by $a_2$

\item \emph{Agent preference maximum on the right} (\rAMR): 
$$
\frac{\Gamma \to \Delta, b \lap_p a}{\Gamma \to \Delta,\forall x. \, x \lap_p a}\; b
$$

\item \emph{Agent preference maximum on the left} (\rAML): 
$$
\frac{a' \lap_p a, \forall x. \, x \lap_p a,\Gamma \to \Delta}{\forall x. \, x \lap_p a, \Gamma \to \Delta}
$$

\item \emph{Negative statement} (\rNB): 
$$
\frac{a: {-}\phi,\Gamma \to \Delta, a: \phi}{a: {-} \phi,\Gamma \to \Delta}
$$

\item \emph{Negative knowledge} (\rNK): 
$$
\frac{{-} \phi,\Gamma \to \Delta, \phi}{{-}  \phi,\Gamma \to \Delta}
$$

\item \emph{Statement propagation over time} (\rBP): 
$$
\frac{\Gamma \to \Delta, t_1 \cong t_2, a:{-} (t_1 \cdot p), a: t_2 \cdot p}{\Gamma \to \Delta, t_1 \cong t_2, a:{-} (t_1 \cdot p)}
$$

\item \emph{Knowledge propagation time} (\rKP): 
$$
\frac{\Gamma \to \Delta, t_1 \cong t_2, {-}(t_1 \cdot p), t_2 \cdot p}{\Gamma \to \Delta, t_1 \cong t_2, {-}(t_1 \cdot p)}
$$

\item \emph{Statement confirmation on the right} (\rBCR): 
$$
\frac{b:{-}\phi, a \lap_{\var_P(\phi)} b, \Gamma \to \Delta}{\Gamma \to \Delta,a :\lnecb \phi} \; b
$$

\item \emph{Statement confirmation on the left} (\rBCL): 
$$
\frac{a :\lnecb \phi,\Gamma \to \Delta, a \lap_{\var_P(\phi)} a' \qquad a :\lnecb \phi,\Gamma \to \Delta, a':{-}\phi}{a :\lnecb \phi,\Gamma \to \Delta}
$$

\item \emph{Extracting knowledge from  statement on the right} (\rBKR): 
$$
\frac{b:{-} \phi, b:\lnecb {-}\phi, \Gamma \to \Delta \quad \Gamma \to \Delta, \phi, a:\phi \quad \Gamma \to \Delta, \phi, a:\lnecb  \phi}{\Gamma \to \Delta, \phi} \; b
$$

\item \emph{Extracting knowledge  on the left} (\rKL): 
$$
\frac{b:\phi, b:\lnecb \phi, \Gamma \to \Delta}{\phi,\Gamma \to \Delta} \; b
$$

\item \emph{Extracting knowledge from statement on the left} (\rBKL): 
$$
\frac{\phi,\Gamma \to \Delta,  a:{-} \phi \qquad  \phi,\Gamma \to \Delta, a:\lnecb {-}\phi}{\phi,\Gamma \to \Delta}
$$

\item \emph{Conditional reasoning on the right} (\CR1): 
$$
\frac{\phi_1,\dots,\phi_n, \Gamma \to \Delta, a : \phi}{\Gamma \to \Delta, a ::(\phi_1,\, \dots, \, \phi_n /  \phi)}
$$

\item \emph{Conditional reasoning on the left} (\CR2): 
$$
\frac{\Gamma \to \Delta, 
\phi_1\quad \cdots \quad  \Gamma \to \Delta, 
\phi_n \quad a : \phi, \Gamma \to \Delta}{a ::(\phi_1,\, \dots, \, \phi_n /  \phi),\Gamma \to \Delta}
$$
\end{itemize}
\end{definition}

Most of the rules are self-explanatory, so let us briefly explain only some of them.

The rule $(\rNB)$ states that the claims of an agent are not contradictory, in the sense that an agent cannot claim $\phi$ and $-\phi$. Observe that we do not have the right counterpart of $\rNB$ that introduces $-$ in the right hand side. This is because the logic is not bivalent since there is a third truth-value $\frac 12$. Therefore, $-$ is not a classical negation.

The rule $(\rBP)$ states that in order to conclude that an agent does not state a certain claim at a particular time it is enough to show that he states the opposite claim at a different time. 

The rule $(\rBCR)$ states that it is possible to conclude $a:\lnecb \phi$ when there is no agent more trustworthy than $a$ that claims $-\phi$. 

On the other hand, rule $(\rBKR)$ states that in order to conclude $\phi$, it is enough to show that there is an agent stating $\phi$ that is not contradicted by a more trustworthy agent and that each agent stating the opposite claim is contradicted by a more trustworthy agent.

We say that a sequent $\Gamma \to \Delta$ is a \emph{theorem}, written  
$$
\der \Gamma \to \Delta\,,
$$ 
if there is a finite sequence of sequents
$$
\Gamma_1 \to \Delta_1\cdots \Gamma_n \to \Delta_n
$$
such that:
\begin{itemize}
\item $\Gamma_1 \to \Delta_1$ is $\Gamma \to \Delta$,
\item for each $i$,  
\begin{itemize}
\item either $\Gamma_i \to \Delta_i$ is the conclusion of a rule without premises,
\item or $\Gamma_i \to \Delta_i$ is the conclusion of a rule where each premise  is a sequent $\Gamma_j\to \Delta_j$ in the sequence  with $j>i$.
\end{itemize}
\end{itemize}
In this case, the sequence is said to be a {\it derivation} for $\Gamma \to \Delta$ and $\Gamma \to \Delta$ is  said to be \emph{derivable}.

The notion of derivable can be brought to the realm of formulas. 
For $\Psi \cup \{\alpha\} \subseteq L_P$, we say that $\alpha$ is \emph{derivable} from $\Psi$, denoted by
$$
\Psi \der \alpha\,,
$$
whenever there is a finite set $\Gamma \subseteq \Psi$ such that
$\der \Gamma \to \alpha$.

For the sake of readability, in the derivations that we give below we underline the principal formula(s) of the rule/axiom that is applied.


\begin{example}\em
Let us return to the crime situation described in Example~\ref{example:one}. Let $P$ be the set composed by the following propositional symbols:
$\inplace$, $\capable$, $\motive$, $\culprit$, $\divergence$, $\poor$, $\friend$, $\guilty$ and $\similar.$
Let $\Gamma$ be the set containing the assertions
\begin{itemize}
\item $\psi_1=a_1::(\Jinplace,\Jcapable\; /\; \Jculprit)$ 
\item $\psi_2=a_2::(\Jmotive\; /\; \Jculprit)$ 
\item $\psi_3=a_3::({-} (\Jcapable),\Jpoor\; /\; {-} (\Jculprit))$ 
\item $\psi_4=a_4::(\Jinplace,\Jfriend\; /\; {-} (\Jculprit))$ 
\item $\psi_5=a_5::(\Jguilty,\Csimilar\; /\; \Jcapable)$ 
\item $\psi_6=a_6::(\Jguilty,{-}(\Csimilar)\; /\; {-}(\Jcapable))$ 
\item $\psi_7=a_7::(\Jdivergence\; /\; \Jmotive)$
\item $\psi_8=a_8: \Jdivergence$
\item $\delta_1=\forall x. \, x \lapmot a_7$
\item $\delta_2=\forall x. \, x \lapdiv a_8$
\end{itemize}
and let $\Delta$ be the singleton set with the assertion $a_2:\Jculprit$. The  derivation in Figure~\ref{fig:der1} establishes that $\Gamma \to \Delta$ is derivable. The subderivation of $(*)$ is in Figure~\ref{fig:der2}.
\begin{figure}[h]
\sder{
\hline
\\[-2mm]
1 \ $\psi_1,\underline{\psi_2},\dots,\psi_8,\delta_1,\delta_2 \to a_2:\Jculprit$ & \CR2:2,3 \\[1mm]
2 \ $\psi_1,\psi_3,\dots,\underline{\psi_7},\psi_8,\delta_1,\delta_2 \to a_2:\Jculprit,\Jmotive$ &  \CR2:4,5\\[1mm]
3 \ $\underline{a_2:\Jculprit}, \psi_1,\psi_3,\dots,\psi_8,\delta_1,\delta_2 \to \underline{a_2:\Jculprit}$ & \Ax\\[1mm]
4 \ $\psi_1,\psi_3,\dots,\psi_6,\psi_8,\delta_1,\delta_2 \to a_2:\Jculprit,\Jmotive,\underline{\Jdivergence}$ & \rBKR:6,7,8 \\[1mm]
5 \ $a_7:\Jmotive,\psi_1,\psi_3,\dots,\psi_6,\psi_8,\delta_1,\delta_2 \to a_2:\Jculprit,\underline{\Jmotive}$ & \rBKR (*) \\[1mm]
6 \ $b: {-}(\Jdivergence), \underline{b:\lnecb {-}(\Jdivergence)},\psi_1,\psi_3,\dots,\psi_6,\psi_8,\delta_1,\delta_2$ \\
\hspace*{10mm} $\to a_2:\Jculprit,\Jmotive$ &  \rBCL:12,13\\[1mm]
7 \ $\psi_1,\psi_3,\dots,\psi_6,\underline{\psi_8},\delta_1,\delta_2 \to a_2:\Jculprit,\Jmotive,\Jdivergence,$\\
\hspace*{10mm} $\underline{a_8:\Jdivergence}$ & \Ax  \\[1mm]
8 \ $\psi_1,\psi_3,\dots,\psi_6,\psi_8,\delta_1,\delta_2 \to a_2::\Jculprit,\Jmotive,\Jdivergence,$\\
\hspace*{10mm} $\underline{a_8:\lnecb\, \Jdivergence}$ & \rBCR:9 \\[1mm]
9 \ $\underline{b: {-}(\Jdivergence)},a_8:\lapdiv b,\psi_1,\psi_3,\dots,\psi_6,\psi_8,\delta_1,\delta_2$\\
\hspace*{10mm} $\to a_2:\Jculprit,\Jmotive,\Jdivergence$ & \rNB:10 \\[1mm]
10 \ $b: {-}(\Jdivergence),a_8:\lapdiv b,\psi_1,\psi_3,\dots,\psi_6,\psi_8,\delta_1,\underline{\delta_2}\to $\\
\hspace*{10mm} $a_2:\Jculprit,\Jmotive,\Jdivergence, b: \Jdivergence$ & \rAML:11 \\[1mm]
11 \ $\underline{b \lapdiv a_8},b: {-}(\Jdivergence),\underline{a_8:\lapdiv b},$\\
\hspace*{10mm} $\psi_1,\psi_3,\dots,\psi_6,\underline{\psi_8},\delta_1,\delta_2\to a_2:\Jculprit,\Jmotive,$\\
\hspace*{10mm} $\Jdivergence,\underline{b: \Jdivergence}$ & \rAC \\[1mm]
12 \ $b: {-}(\Jdivergence), b:\lnecb {-}(\Jdivergence),\psi_1,\psi_3,\dots,\psi_6,\psi_8,\delta_1,\underline{\delta_2}$\\
\hspace*{10mm} $ \to a_2:\Jculprit,\Jmotive,b\lapdiv a_8$ & \rAML:14 \\[1mm]
13 \ $b: {-}(\Jdivergence), b:\lnecb {-}(\Jdivergence),\psi_1,\psi_3,\dots,\psi_6,\underline{\psi_8},\delta_1,\delta_2$\\
\hspace*{10mm} $ \to a_2:\Jculprit,\Jmotive,\underline{a_8:\Jdivergence}$ &\Ax  \\[1mm]
14 \ $\underline{b\lapdiv a_8},b: {-}(\Jdivergence),b:\lnecb {-}(\Jdivergence),$\\
\hspace*{10mm} $\psi_1,\psi_3,\dots,\psi_6,\psi_8,\delta_1,\delta_2\to a_2:\Jculprit,\Jmotive,\underline{b\lapdiv a_8}$ & \Ax 
\\[4mm]\hline}
\caption{Derivation of $\psi_1,{\psi_2},\dots,\psi_8,\delta_1,\delta_2 \to a_2:\Jculprit$}\label{fig:der1}
\end{figure}
\noindent
\begin{figure}[h]
\sder{\hline
\\[-2mm]
1 \ $a_7:\Jmotive,\psi_1,\psi_3,\dots,\psi_6,\psi_8,\delta_1,\delta_2 \to a_2:\Jculprit,\underline{\Jmotive}$ & \rBKR:2,3,4 \\[2mm]
2 \ $b: {-}(\Jmotive), \underline{b:\lnecb {-}(\Jmotive)},a_7:\Jmotive,\psi_1,\psi_3,\dots,\psi_6,$ \\
\hspace*{10mm} $\psi_8,\delta_1,\delta_2 \to a_2:\Jculprit$ &  \rBCL:5,6\\[2mm]
3 \ $\underline{a_7:\Jmotive},\psi_1,\psi_3,\dots,\psi_6,\psi_8,\delta_1,\delta_2 \to a_2:\Jculprit,\Jmotive,$ \\
\hspace*{10mm} $\underline{a_7:\Jmotive}$ & \Ax \\[2mm]
4 \ $a_7:\Jmotive,\psi_1,\psi_3,\dots,\psi_6,\psi_8,\delta_1,\delta_2 \to a_2:\Jculprit,\Jmotive,$ \\
\hspace*{10mm} $\underline{a_7:\lnecb\, \Jmotive}$ & \rBCR:8 \\[2mm]
5 \ $b: {-}(\Jmotive), b:\lnecb {-}(\Jmotive),a_7:\Jmotive,\psi_1,\psi_3,\dots,\psi_6,$ \\
\hspace*{10mm} $\psi_8,\underline{\delta_1},\delta_2 \to a_2:\Jculprit,b\lapmot a_7$ & \rAML:7 \\[2mm]
6 \ $b: {-}(\Jmotive), b:\lnecb {-}(\Jmotive),\underline{a_7:\Jmotive},\psi_1,\psi_3,\dots,\psi_6,$ \\
\hspace*{10mm} $\psi_8,\delta_1,\delta_2 \to a_2:\Jculprit,\underline{a_7:\Jmotive}$ &\Ax  \\[2mm]
7 \ $\underline{b\lapmot a_7},b: {-}(\Jmotive), b:\lnecb {-}(\Jmotive),a_7:\Jmotive,$ \\
\hspace*{10mm} $\psi_1,\psi_3,\dots,\psi_6,\psi_8,\delta_1,\delta_2 \to a_2:\Jculprit,\underline{b\lapmot a_7}$ & \Ax \\[2mm]
8 \ $\underline{b:{-}(\Jmotive)},a_7\lapdiv b,a_7:\Jmotive,\psi_1,\psi_3,\dots,\psi_6,$ \\
\hspace*{10mm} $\psi_8,\delta_1,\delta_2 \to a_2:\Jculprit,\Jmotive$ & \rNB:9 \\[2mm]
9 \ $b:{-}(\Jmotive),a_7\lapdiv b,a_7:\Jmotive,\psi_1,\psi_3,\dots,\psi_6,$ \\
\hspace*{10mm} $\psi_8,\underline{\delta_1},\delta_2 \to a_2:\Jculprit,\Jmotive,b:\Jmotive$ & \rAML:10 \\[2mm]
10 \ $\underline{b\lapdiv a_7},b:{-}(\Jmotive),\underline{a_7\lapdiv b},\underline{a_7:\Jmotive},$ \\
\hspace*{10mm} $\psi_1,\psi_3,\dots,\psi_6,\psi_8,\underline{\delta_1},\delta_2 \to a_2:\Jculprit,\Jmotive,\underline{b:\Jmotive}$ & \rAC
\\[4mm]\hline}
\caption{Subderivation of $(*)$ in the derivation in Figure~\ref{fig:der1}}\label{fig:der2}
\end{figure}
\noindent
\hfill $\blacksquare$
\end{example}

We will show the use of Time-Stamped Claim Logic with a more complex case study in Section~\ref{sec:case-study}.

\section{The Semantics of  Time-Stamped Claim Logic}\label{sec:semantics}

In this section, we define the semantics of the Time-Stamped Claim Logic. In the following sections, we will then prove the soundness and completeness of the sequent calculus with respect to this semantics. We start by introducing the notion of interpretation structure, which has a modal and many-valued flavour. 

\begin{definition}
An \emph{interpretation structure} over a non-empty set $P$ of propositional symbols is a tuple
$$
(D_A,D_T,\cong^I,\{\lap_p^I\}_{p \in P},V)
$$
such that
\begin{itemize}
\item $D_A$ and $D_T$ are non-empty sets,
\item $\cong^I$ is a reflexive, symmetric and transitive relation over $D_T$,
\item $\lap_p^I$ is a transitive and reflexive binary relation over $D_A$ for each $p \in P$,
\item $V: P \times D_A  \times D_T \to \{0,1,\frac12\}$ such that for every $p \in P$, $d \in D_A$ and $n,n' \in D_T$,
\begin{itemize}
\item if $V(p,d,n)=1$, then $V(p,d,m)=0$ for every $m \in D_T$ such that $m \not \cong^I n$,
\item $V(p,d,n) = V(p,d,n')$ whenever $n \cong^I n'$,
\item $V(p,d,n) = V(p,d',n)$ whenever $d  \lap_p^I d'$ and $d'  \lap_p^I d$.
\end{itemize}
\end{itemize}
In an interpretation structure, $V(p,d,n)=1$ means that agent $d$ claims that $p$ occurred at time $n$, $V(p,d,n)=0$ means that agent $d$ claims that $p$ does not occur at time $n$ and $V(p,d,n)=\frac 12$ means that agent $d$ does not have an opinion about the occurrence of $p$ at time  $n$.

An \emph{assignment} $\rho$ over $I$ is a pair $(\rho_A,\rho_T)$ such that $\rho_A: X_A \to D_A$ and $\rho_T: X_T \to D_T$ are maps. Moreover, we say that  assignments $\rho$ and $\rho'$ are \emph{equivalent} up to $b \in X_A$, written 
$$
\rho \equiv^A_b \rho'\,,
$$ 
whenever $\rho'_T=\rho_T$ and $\rho'_A(a)=\rho_A(a)$ for every $a \in X_A \setminus \{b\}$. 

Given an interpretation structure $I$ over $P$ and an assignment $\rho$ over $I$, \emph{satisfaction of an assertion}  $\alpha$ by $I$ and $\rho$, denoted by
$$
I \rho \sat \alpha\,,
$$
is defined as follows
\begin{itemize}
\item $I \rho \sat t_1 \cong t_2$ whenever $\rho_T(t_1) {\cong}^I \rho_T(t_2)$,
\item $I \rho \sat a_1 \lap_{p} a_2$ whenever $\rho_A(a_1) \lap_p^I \, \rho_A(a_2)$,
\item $I \rho \sat \forall x. \, x \lap_p a$ whenever $d \lap_p^I \rho_A(a)$ for every $d \in D_A$,
\item $I \rho \sat a : \phi$ whenever $V(p,\rho_A(a),\rho_T(t))= 
\begin{cases}
1 & \text{if } \phi \text{ is } t \cdot p\\
0 & \text{if } \phi \text{ is } {-}(t \cdot p)\,,
\end{cases}
$
\item $I \rho \sat a:\lnecb \phi$ whenever, for each assignment $\rho'$ and $b \in X_A$ such that $\rho' \equiv^A_b \rho$, if $I \rho' \sat a \lap_{\var_P(\phi)} b$ then $I \rho' \not \sat  b: {-} \phi$,
\item $I \rho \sat \phi$ whenever 
\begin{itemize}
\item for each assignment $\rho'$ and $b \in X_A$ with $\rho' \equiv^A_b \rho$, if $I \rho' \sat b:{-} \phi$ then $I \rho' \not \sat b:\lnecb {-} \phi$,
\item there are assignment $\rho'$ and  $a \in X_A$ with $\rho' \equiv^A_a \rho$, $I \rho' \sat a : \phi$ and $I \rho' \sat a : \lnecb \phi$,
\end{itemize}
\item $I \rho \sat a ::(\phi_1,\, \dots, \, \phi_n / \phi)$ whenever if $I \rho \sat \phi_i$ for $i=1,\dots,n$ then $I \rho \sat a: \phi$.
\hfill $\blacksquare$
\end{itemize}
\end{definition}

Note that in the definition of satisfaction the truth-value $\frac 12$ does not play an explicit role. However, it allows the following to be possible:
$$
I \rho \not\sat a : \phi \quad \text{and} \quad I \rho \not \sat a : -\phi.
$$

To simplify the presentation, given a set $\Gamma \subseteq L_P$, we will write
$$
I \rho \sat \Gamma
$$ 
whenever $I \rho \sat \gamma$ for each $\gamma \in \Gamma$. 
\begin{definition}
Given $\Psi \subseteq L_P$ and $\alpha \in L_P$, we say that $\Psi$ \emph{entails} $\alpha$, denoted by
$$
\Psi \ent \alpha\,,
$$
whenever for every interpretation structure $I$ and assignment $\rho$ over $I$ we have 
$$
I \rho \sat \alpha \quad \text{ if }  \quad I \rho \sat \Psi\,.
$$
\hfill $\blacksquare$
\end{definition}

\section{Soundness}\label{sec:soundness}

In this section, we prove that the theorems derived using the calculus of the Time-Stamped Claim Logic introduced in Section~\ref{sec:calculus} are valid according to the semantics introduced in Section~\ref{sec:semantics}. 
We start by providing some relevant semantic notions related to sequents. 

\begin{definition}
We say that an interpretation structure $I$ and an assignment $\rho$ \emph{satisfy a sequent} $\Gamma \to \Delta$, denoted by
$$
I \rho \sat \Gamma \to \Delta\,,
$$ 
whenever $I\rho \sat \delta$ for some $\delta \in \Delta$ if $I\rho \sat \Gamma$. We say that a sequent $\Gamma \to \Delta$ is \emph{valid}, denoted by
$$
\ent \Gamma \to \Delta\,,
$$ 
whenever it is satisfied by every interpretation structure $I$ and assignment $\rho$. 
We say that $I$ and $\rho$ \emph{satisfy a rule with a fresh variable} $b$, 
$$
\frac{\Gamma_1 \to \Delta_1 \quad \cdots \quad \Gamma_n \to \Delta_n}{\Gamma \to \Delta}\;  b \,,
$$
whenever $I \rho \sat \Gamma \to \Delta$ if $I \rho' \sat \Gamma_i \to \Delta_i$ for $i=1,\dots,n$
for every $\rho'$ such that $\rho' \equiv^A_b \rho$. We say that $I$ and $\rho$ \emph{satisfy a rule without fresh variables},
$$
\frac{\Gamma_1 \to \Delta_1 \quad \cdots \quad \Gamma_n \to \Delta_n}{\Gamma \to \Delta}\,,
$$
whenever $I \rho \sat \Gamma \to \Delta$ if $I \rho \sat \Gamma_i \to \Delta_i$ for $i=1,\dots,n$. 
Finally, we say that a rule is \emph{sound} whenever it is satisfied by every interpretation structure and assignment. 
\hfill $\blacksquare$
\end{definition}

The proof of the following result is in Appendix~A.

\begin{prop}\label{prop:rulessound}
The rules of the sequent calculus are sound.
\end{prop}

The soundness of the rules can be used to show the soundness of deductive consequences.

\begin{prop}\label{prop:soundnessder}
Let $\Gamma \to \Delta$ be a sequent over $P$. Then $\ent \Gamma \to \Delta$ whenever $\der \Gamma \to \Delta$.
\end{prop}
\begin{proof}
The proof is by induction on the length of a derivation $\Gamma_1 \to \Delta_1 \cdots \Gamma_n \to \Delta_n$ of $\Gamma \to \Delta$. The base case follows straightforwardly by Proposition~\ref{prop:rulessound}.

For the step case, assume that $\Gamma_1 \to \Delta_1$ follows from rule $u$ of the calculus applied to sequents 
$\Gamma_{i_j} \to \Delta_{i_j}$ for $j=1,\dots, k$. Then
 $\der \Gamma_{i_j} \to \Delta_{i_j}$ for $j=1,\dots, k$. Hence, by the induction hypothesis, $\ent \Gamma_{i_j} \to \Delta_{i_j}$ for $j=1,\dots, k$. 
 Let $I$ be an interpretation structure, $\rho$ an assignment and $\rho'$ an assignment equivalent modulo the fresh variables of $u$ to $\rho$. Then $I \rho' \sat \Gamma_{i_j} \to \Delta_{i_j}$ for $j=1,\dots, k$ since $\ent \Gamma_{i_j} \to \Delta_{i_j}$ for $j=1,\dots, k$. Taking into account that $u$ is sound, by Proposition~\ref{prop:rulessound}, we conclude that 
 $I \rho \sat \Gamma_1 \to \Delta_1$.
\end{proof}

Capitalizing on the soundness at the sequent level, we can establish soundness at the level of assertions.

\begin{prop}
Let $\Psi \subseteq L_P$ and $\alpha \in L_P$. Then $\Psi \ent \alpha$ whenever $\Psi \der \alpha$.
\end{prop}
\begin{proof}
Assume that $\Psi \der \alpha$. Then there is a finite set $\Gamma \subseteq \Psi$ such that
$\der \Gamma \to \alpha$. Thus, by Proposition~\ref{prop:soundnessder}, $\ent \Gamma \to \alpha$. Let $I$ be an interpretation structure and $\rho$ an assignment such that $I \rho \sat \Psi$. Therefore,  $I \rho \sat \Gamma$
and so $I \rho \sat \alpha$.
\end{proof}

\section{Completeness}\label{sec:completeness}

In this section, we show that the Time-Stamped Claim Logic is complete using a Hintikka style technique. We consider Hintikka pairs instead of just Hintikka sets since the logic does not have  a classical negation over all assertions. 

\begin{definition}
A \emph{Hintikka pair over} $P' \subseteq P$, $X'_A \subseteq X_A$ and $X'_T \subseteq X_T$  is a pair $(\Psi_L,\Psi_R)$ of sets of assertions in $L_{P'}^{X'_A,X'_T}$ such that:
\begin{itemize}
\item either $\beta \notin \Psi_L$ or $\beta \notin \Psi_R$ for each  $\beta \in L_{P'}^{X'_A,X'_T}$ of the form $a: \phi$, $t_1 \cong t_2$ and  $a_1 \lap_p a_2$,
\item either $\beta \in \Psi_L$ or $\beta \in \Psi_R$ for each $\beta \in L_{P'}^{X'_A,X'_T}$  of the form $a: \phi$, $t_1 \cong t_2$ and $a_1 \lap_p a_2$,
\item $t \cong t \notin \Psi_R$ and similarly for $\lap_p$,
\item if $t_1 \cong t_2 \in \Psi_L$ then $t_2 \cong t_1 \notin \Psi_R$,
\item if $t_1 \cong t_2, t_2 \cong t_3 \in \Psi_L$ then $t_1 \cong t_3 \notin \Psi_R$ and similarly for  $\lap_p$,
\item if $t_1 \cong t_2,[a:\phi]^{t_1}_{t_2} \in \Psi_L$ then $a:\phi \notin \Psi_R$,
\item if $a_1 \lap_p a_2,a_2 \lap_p a_1, [a:\phi]^{a_1}_{a_2} \in \Psi_L$ then $a:\phi \notin \Psi_R$,
\item if $\forall x. \, x \lap_p a \in \Psi_R$ then there is  $a' \in X'_A$ such that $a'\lap_p a \in \Psi_R$,
\item if $\forall x. \, x \lap_p a \in \Psi_L$ then $a_1 \lap_p a \in \Psi_L$ for each $a_1 \in X'_A$,
\item if $a: \phi \in \Psi_L$ then $a: {-} \phi \in \Psi_R$,
\item if $t_1 \cong t_2, a: {-} (t_1 \cdot p) \in \Psi_R$ then $a:t_2 \cdot p \in \Psi_R$,
\item if $a: \lnecb \phi \in \Psi_R$ then there is $a' \in X'_A$ such that $a \lap_{\var(\phi)} a', a': {-} \phi \in \Psi_L$,
\item if $a: \lnecb \phi \in \Psi_L$ then, for each $a' \in X'_A$ either $a \lap_{\var(\phi)} a' \in \Psi_R$ or $a': {-} \phi \in \Psi_R$,  
\item if $\phi \in \Psi_R$ then  either there exists $a' \in X'_A$ such that $a':{-} \phi, a':\lnecb {-}\phi \in \Psi_L$ or for each $a \in X'_A$  either $a:\phi \in \Psi_R$ or $a:\lnecb \phi  \in \Psi_R$,
\item if $\phi \in \Psi_L$ then there is $a' \in X'_A$ such that $a':\phi, a':\lnecb \phi \in \Psi_L$ and for each $a \in X'_A$ either $a:{-} \phi \in \Psi_R$ or $a:\lnecb{-} \phi \in \Psi_R$,
\item if $a::(\phi_1,\dots,\phi_n / \phi) \in \Psi_R$ then if $\phi_1,\dots,\phi_n \in \Psi_L$ then	$a:\phi \in \Psi_R$,
\item if $a::(\phi_1,\dots,\phi_n / \phi) \in \Psi_L$ then either $\phi_i\in \Psi_R$ for some $i=1,\dots,n$ or $a:\phi \in \Psi_L$.
\hfill $\blacksquare$
\end{itemize}
\end{definition}

We now show that for each Hintikka pair there are an interpretation structure and an assignment satisfying all the assertions in the first component of the pair and not satisfying each assertion in the second component. 

The proof of the next result can be found in Appendix~A.

\begin{prop}\label{prop:HKsat}
Let $H=(\Psi_L,\Psi_R)$ be a Hintikka pair over $P'$, $X'_A$ and $X'_T$. Then there are an interpretation structure $I_H$  over $P'$ and an assignment $\rho$ over $I_H$ such that $I_H \rho \sat \psi$ for each $\psi \in \Psi_L$ and $I_H \rho \not \sat \psi$
for each $\psi \in \Psi_R$. 
\end{prop}

The main objective now is to show that if a sequent is not derivable then there is an interpretation structure that falsifies it. For this purpose, we must introduce the concept of (deductive) expansion of a sequent.
\begin{definition}
An \emph{expansion} of a sequent $\Gamma \to \Delta$ is a  (finite or infinite) sequence of sequents $\Gamma_1 \to \Delta_1 \ \cdots$ such that:
\begin{itemize}
\item $\Gamma_1 \to \Delta_1$ is $\Gamma \to \Delta$;
\item for each $i\neq 1$, $\Gamma_i \to \Delta_i$ is a premise of a rule with conclusion $\Gamma_j\to \Delta_j$ with $j <i$ and all the other premises of the rule occur in the sequence. 
\end{itemize}
\hfill $\blacksquare$
\end{definition}

Observe that every derivation is an expansion but not vice versa. From an expansion fulfilling conditions described below, we will define an interpretation structure that falsifies the first sequent in the expansion. 
We need first to introduce the concept of branch.
\begin{definition}
A \emph{branch} of the expansion $\Gamma_1 \to \Delta_1 \  \cdots$ starting at sequent $\Gamma_i \to \Delta_i$ is a subsequence 
$\Gamma_{i_1} \to \Delta_{i_1} \ \cdots$ of the expansion such that:
\begin{itemize}
\item $\Gamma_{i_1} \to \Delta_{i_1}$ is $\Gamma_i \to \Delta_i$,
\item for each $j >1$, $\Gamma_{i_j} \to \Delta_{i_j}$ is a premise of the rule used in the expansion to justify $\Gamma_{i_{j-1}} \to \Delta_{i_{j-1}}$ (the conclusion of the rule),
\item for each $j \geq 1$ if $\Gamma_{i_j} \to \Delta_{i_j}$ is the conclusion of a rule in the expansion then there $k >j$ such $\Gamma_{i_k} \to \Delta_{i_k}$ is a premise of the rule in the expansion.
\end{itemize}
For simplicity, we present a branch as a sequence  $\Gamma'_1 \to \Delta'_1 \ \Gamma'_2 \to \Delta'_2 \ \cdots$ instead of using  the numbering of each sequent of the branch in the expansion. When $\Gamma'_1 \to \Delta'_1$ is $\Gamma_1 \to \Delta_1$ we say that the branch is \emph{rooted}. 
\hfill $\blacksquare$
\end{definition}

\begin{definition}
We say that a finite branch $\Gamma'_1 \to \Delta'_1 \cdots\Gamma'_n \to \Delta'_n$ of an expansion $\Gamma_1 \to \Delta_1 \ \cdots$ is \emph{analytical} whenever it is possible to inductively define the families
of sets
$$
\{(i_L)_j\}_{j=1,\dots n} \quad \text{and} \quad \{(i_R)_j\}_{j=1,\dots,n}
$$  
such that 
$$
(i_L)_1=\Gamma'_1 \quad \text{and} \quad (i_R)_1=\Delta'_1
$$ 
and for $j=2,\dots,n$,  
$$
(i_L)_j=(i_L)_{j-1} \cup (\Gamma'_j \setminus \Gamma'_{j-1}) \quad \text{and} \quad 
(i_R)_j=(i_R)_{j-1} \cup (\Delta'_j \setminus \Delta'_{j-1})
$$
and  if  $\Gamma'_j \to \Delta'_j$ is a premise of a rule $u$,
then
\begin{enumerate}
\item either there is a side formula in the left hand side of a premise of $u$ in the expansion which does not occur in $(i_L)_{j-1}$ not involving fresh variables or there is a side formula in the right hand side of a premise of $u$ in the expansion which does not occur in  $(i_R)_{j-1}$ not involving fresh variables; 
\item the non fresh variables of $u$ in $\var_A((\Gamma'_j \setminus \Gamma'_{j-1}) \cup (\Delta'_j \setminus \Delta'_{j-1}))$ are in  $\var_A((i_L)_{j-1} \cup (i_R)_{j-1})$ and similarly for  $\var_T$ and $\var_P$; 
\item the fresh variables of $u$ in $\var_A((\Gamma'_j \setminus \Gamma'_{j-1}) \cup (\Delta'_j \setminus \Delta'_{j-1}))$ are not in $\var_A(\Gamma'_{j-1} \cup \Delta'_{j-1})$ and similarly for  $\var_T$ and $\var_P$. 
\end{enumerate}
An analytical branch $\Gamma'_1 \to \Delta'_1 \cdots\Gamma'_n \to \Delta'_n$ is \emph{exhausted} whenever each branch that results from applying a rule to  $\Gamma'_n \to \Delta'_n$ is not analytical in the augmented expansion. An analytical exhausted branch is \emph{closed} whenever the last sequent is the conclusion of a rule without premises. Otherwise, it is said to be \emph{open}. An expansion is
\emph{open} when it has an open analytical exhausted branch. 
\hfill $\blacksquare$
\end{definition}

The proof of the next result is in Appendix~A.

\begin{prop}\label{prop:Hintikkafrombranch} 
Let $\Gamma'_1 \to \Delta'_1 \cdots\Gamma'_n \to \Delta'_n$ be a rooted  open analytical exhausted branch of 
an expansion $\Gamma_1 \to \Delta_1 \ \cdots$. Then $$\left(\displaystyle \bigcup_{i=1}^n \Gamma'_i,\displaystyle \bigcup_{i=1}^n \Delta'_i\right)$$ is an Hintikka pair over $\var_P(\Gamma'_n \cup \Delta'_n)$, $\var_A(\Gamma'_n \cup \Delta'_n)$ and $\var_T(\Gamma'_n \cup \Delta'_n)$.
\end{prop}

\begin{prop}\label{thm:thvsexp} 
If a sequent is not a theorem then it has an expansion with a rooted  open analytical exhausted branch.
\end{prop}
\begin{proof}
It is straightforward to see that every sequent has an expansion where all rooted branches are analytical and exhausted.  Furthermore,  if all the rooted branches in that expansion  are not open then  the expansion is a derivation. 
\end{proof}

Capitalizing on the previous results in this section, we are able to establish the completeness result at the level of sequents.

\begin{prop}\label{thm:sequentcompleteness} 
If a sequent is valid then it is a theorem.
\end{prop}
\begin{proof}
Assume that the sequent $\Gamma \to \Delta$ is not a theorem. Then, by Proposition~\ref{thm:thvsexp}, it has an expansion with a rooted  open analytical exhausted branch. So, by Proposition~\ref{prop:Hintikkafrombranch},
there is an Hintikka pair $H=(\Psi_L,\Psi_R)$ over $\var_P(\Gamma'_n \cup \Delta'_n)$, $\var_A(\Gamma'_n \cup \Delta'_n)$ and $\var_T(\Gamma'_n \cup \Delta'_n)$ such that $\Gamma \subseteq \Psi_L$ and $\Delta \subseteq \Psi_R$. Hence, by Proposition~\ref{prop:HKsat}, 
there are an interpretation structure $I_H$ over $\var_P(\Gamma'_n \cup \Delta'_n)$ and an assignment $\rho$ over $I_H$ such that $I_H \rho \sat \Gamma$ and  $I_H \rho \not \sat \delta$ for each $\delta \in \Delta$. Let $I$ be an interpretation structure over $P$ such that $I_H$ is the reduct of $I$ to $\var_P(\Gamma'_n \cup \Delta'_n)$. Observe that $\rho$ is also an assignment over $I$. 
Hence, $I \rho \sat \Gamma$ and  $I \rho \not \sat \delta$ for each $\delta \in \Delta$.
Thus, $I \rho \not \sat \Gamma \to \Delta$ and so $\Gamma \to \Delta$ is not valid. 
\end{proof}

Finally, we can show that the Time-Stamped Claim Logic is complete.
\begin{prop}
Let $\Gamma \subseteq L_P$ be a finite set and $\alpha \in L_P$. Then $\Gamma \der \alpha$ whenever $\Gamma \ent \alpha$.
\end{prop}
\begin{proof}
Assume that $\Gamma \ent \alpha$. Then $\ent \Gamma \to \alpha$. Thus, 
by Proposition~\ref{thm:sequentcompleteness}, $\der \Gamma \to \alpha$ and so $\Gamma \der \alpha$.
\end{proof}

\section{A Concrete Case Study in Cyber-Attack Attribution}
\label{sec:case-study}
In order to show how Time-Stamped Claim Logic works, we consider a concrete case study inspired by the Ukraine Power Grid Attack that occurred in December 2015~\cite{case2016analysis}. For the sake of simplicity, we removed the references to the real actors involved in the attack and substituted them with generic ones. 
A good part of the provided evidence
will be similar to the one found for this attack --- we just carried out some simplifications for the sake of understandability.

\subsection{The Set-Up}
We will denote the analyzed attack with $\Att$ and its  victim with $\Victim$.
We have 11 agents/sources that provides information about this attack: $\FE$, $\Sans$ and $\SC$ are three different security companies;
$\DHS$ is the security department of an impartial country;
$\R$ and $\ROne$ are two sources with connections with country $\Russia$; 
$\BBC$ is a well-known international news agency;
$\U$ is a non-verified source; 
$\In$ is the forensics investigator hired by the victim;
$\SansIn$ is the forensics investigator of company $\Sans$;
$\Blog$ is an analyst that writes blog posts related to cyber-attacks.
Furthermore, $\Russia$ is a country,
 and $\Sandworm$ denotes a famous group of hackers.

For what concerns time, $t'$ is Spring 2015, $t''$ is December 2015, whereas $t$ is a more general instant of time that can spread from Spring 2015 to (including) December 2015. 

\subsection{The Evidence}

In the following, we list the evidence provided by the analyst. We start with the given evidence and later provide the given derived evidence. 

$\FE$ states
 that the IPs from where the attack was originated are geolocated in $\Russia$ ($\geoSourceIP$).
\begin{displaymath}\small
\renewcommand\arraystretch{1.3}
\begin{array}{l}
\FE : t\cdot\geoSourceIP^{\Russia}_{\Att}\\
\end{array}
\end{displaymath}
$\FE$ and $\Sans$ both state that the used malware ($\usedMalware$) in this attack ($\Att$) was a version of the \textrm{BlackEnergy} malware, $\BEthree$, whereas $\R$ states that the malware used in the attack was not $\BEthree$ but $\KillDisk$. The provided information holds as different attackers can use different malware or pieces of software for the same attack. 
\begin{displaymath}\small
\renewcommand\arraystretch{1.3}
\begin{array}{ll}
\FE: t\cdot\usedMalware^{\BEthree}_{\Att}
 \  \ \qquad  \qquad  & \ \ 
\R: {-} t\cdot\usedMalware^{\BEthree}_{\Att}\\
\Sans: t\cdot\usedMalware^{\BEthree}_{\Att}
\  \ \qquad \qquad  & \ \
\R: t\cdot\usedMalware^{\KillDisk}_{\Att}
\end{array}
\end{displaymath}
Source $\Sans$ states that $\KillDisk$ was used during the attack, as it
found traces of the $\KillDisk$ malware when analysing the infected system.
$\DHS$ and $\FE$ both state that $\BEthree$ is similar to its previous version $\BEtwo$, whereas $\R$ denies this statement. 
\begin{displaymath}\small
\renewcommand\arraystretch{1.3}
\begin{array}{ll}
\Sans: t\cdot\usedMalware^{\KillDisk}_{\Att}
 \  \ \qquad  \qquad  & \ \ 
\FE:  t\cdot\simil^{\BEthree}_{\BEtwo}\\
\DHS: t\cdot\simil^{\BEthree}_{\BEtwo}
 \  \ \qquad  \qquad  & \ \ 
\R: {-}t\cdot\simil^{\BEthree}_{\BEtwo}
\end{array}
\end{displaymath}
Furthermore, at $\DHS$'s website is stated that $\BEtwo$ is used mainly ($\usedBy$) by a group of attackers denoted by $\Sandworm$.
$\FE$ states that $\Sandworm$ targeted in the past mainly victims 
($\targetVictim$) that have shown some form of \emph{opposition} against $\Russia$, its government, or its economical interest, denoted by $\Russia\_\Opponent$. 
\begin{displaymath}\small
\renewcommand\arraystretch{1.3}
\begin{array}{ll}
\DHS: t\cdot\usedBy^{\BEtwo}_{\Sandworm}
 \  \ \qquad  \qquad  & \ \ 
\FE: t\cdot\targetVictim^{\Russia\_\Opponent}_{\Sandworm}
\end{array}
\end{displaymath}
Let us now introduce the evidence related the negative consequences of the attack $\Att$. 
$\R$ states that the victim suffered from economical losses from the attack ($\econLoss$), while the news agency $\BBC$ states that the main losses were in reputation with respect to the citizens ($\repLoss$). Other agents like $\Sans$ states that the damages of $\Att$ were limited ($\limitedDamage$), whereas 
$\Victim$ states that this assertion is not true. 
\begin{displaymath}\small
\renewcommand\arraystretch{1.3}
\begin{array}{ll}
\R: t\cdot\econLoss^{\Att}_{\Victim}
 \  \ \qquad  \qquad  & \ \ 
 \Sans: t\cdot\limitedDamage^{\Att}_{\Victim}\\
\BBC: t\cdot\repLoss^{\Att}_{\Victim}
 \  \ \qquad  \qquad  & \ \ 
\U: {-} t\cdot\limitedDamage^{\Att}_{\Victim}
\end{array}
\end{displaymath}
Further information were provided by the victim's forensics investigator $\In$, which states that he $\found$, in the victim's email server, phishing emails containing $\BEthree$, dated in Spring 2015, denoted by $t'$.
Furthermore, $\In$ states that in $t'$ there was a campaign of phishing emails ($\spearPhishing$)
and the $\infection$ with $\BEthree$ happened through those phishing emails.
$\Sans$ states that $\BEthree$ is able to create a $\backdoor$ to the victim's system. 
The statement about the infection was also confirmed by investigator $\SansIn$, who states also 
that the infection with $\BEthree$ did not happened by exploiting the \emph{Human Machine Interaction} (HMI) vulnerabilities, and this vulnerabilities were not exploited
($\exploitedVuln$). This last statement is provided also by $\In$.
\begin{displaymath}\small
\renewcommand\arraystretch{1.3}
\begin{array}{ll}
 \In: t'\cdot\found^{\BEthree}_{\phishEmails}
  \  \ \qquad  \qquad  & 
   \SansIn: t'\cdot\infection^{\BEthree}_{\phishEmails}
 \\
 
  \In: t'\cdot\spearPhishing^{\phishEmails}_{\Att}
 \  \ \qquad  \qquad  & 
 \SansIn: {-}t'\cdot\infection^{\BEthree}_{\HMIvuln}

\\

\In: t'\cdot\infection^{\BEthree}_{\phishEmails}
 \  \ \qquad  \qquad  & 
  \SansIn: {-}t'\cdot\exploitedVuln^{\Att}_{\HMIvuln}

\\

\Sans: t'\cdot\backdoor^{\BEthree}_{\Victim}
 \  \ \qquad  \qquad  & 
   \In: {-}t'\cdot\exploitedVuln^{\Att}_{\HMIvuln}
\end{array}
\end{displaymath}
There are more information about how $\Victim$ was infected. 
In particular, $\Sans$ states that $\BEtwo$ infection occurs using HMI vulnerabilities, whereas $\ROne$ states that $\BEthree$ infection did not occurred using phishing emails. 
$\Sans$ states that $\KillDisk$ is used by $\Sandworm$, while $\R$ states the opposite,
and $\ROne$ states that $\KillDisk$ is used by $\Russia$. 
The news agency, $\BBC$, states the $\news$ about the blackout of the victim in December 2015, denoted by $t''$,
and that there was a political conflict ($\polConflict$) between the victim and country $\Russia$ in that period. 
\begin{displaymath}\small
\renewcommand\arraystretch{1.3}
\begin{array}{ll}
\Sans: t'\cdot\infection^{\BEtwo}_{\HMIvuln}
 \  \ \qquad  \qquad  & \ \ 
\ROne: t\cdot\usedBy^{\KillDisk}_{\Russia}\\

\ROne: {-}t'\cdot\infection^{\BEthree}_{\phishEmails}
 \  \ \qquad  \qquad  & \ \ 

\BBC: t''\cdot\news^{\Blackout}_{\Victim}\\

\Sans: t\cdot\usedBy^{\KillDisk}_{\Sandworm}
 \  \ \qquad  \qquad  & \ \ 
\BBC: t\cdot\polConflict^{\Russia}_{\Victim}\\

\R: {-}t\cdot\usedBy^{\KillDisk}_{\Sandworm}
 \  \ \qquad  \qquad  & \ \ 
\\

\end{array}
\end{displaymath}

\begin{figure}[t]
\sder{
\hline
\\[-2mm]
1 \ $\psi_1,\delta_2 \to 
\underline{t \cdot\limitedDamage^{\Att}_{\Victim}}$ &   \rBKR 2,3,4 \\[2mm]
2 \ $b:{-} t \cdot\limitedDamage^{\Att}_{\Victim}, \underline{b:\lnecb {-} t \cdot\limitedDamage^{\Att}_{\Victim}}, \psi_1,\delta_2 \to$\\
\hspace*{10mm} $ 
t \cdot\limitedDamage^{\Att}_{\Victim}$ & \rBCL:5,6  \\[2mm]
3 \ $\underline{\psi_1},\delta_2 \to 
t \cdot\limitedDamage^{\Att}_{\Victim},$\\
\hspace*{10mm} $\underline{\Sans: t \cdot\limitedDamage^{\Att}_{\Victim}}$ & \Ax \\[2mm]
4 \ $\psi_1,\delta_2 \to t \cdot\limitedDamage^{\Att}_{\Victim},$\\
\hspace*{10mm} $\underline{\Sans:\lnecb t \cdot\limitedDamage^{\Att}_{\Victim}}$ & \rBCR:8 \\[2mm]
5 \ $b:{-} t \cdot\limitedDamage^{\Att}_{\Victim}, b:\lnecb {-} t \cdot\limitedDamage^{\Att}_{\Victim}, \psi_1,\underline{\delta_2} \to$\\
\hspace*{10mm} $ 
t \cdot\limitedDamage^{\Att}_{\Victim},b \lap_{\limitedDamage} \Sans$ & \rAML:7  \\[2mm]
6 \ $b:{-} t \cdot\limitedDamage^{\Att}_{\Victim}, b:\lnecb {-} t \cdot\limitedDamage^{\Att}_{\Victim}, \underline{\psi_1},\delta_2 \to$\\
\hspace*{10mm} $ 
t \cdot\limitedDamage^{\Att}_{\Victim}, \underline{\Sans:t \cdot\limitedDamage^{\Att}_{\Victim}}$ & \Ax  \\[2mm]
7 \ $\underline{b \lap_{\limitedDamage} \Sans}, b:{-} t \cdot\limitedDamage^{\Att}_{\Victim}, b:\lnecb {-} t \cdot\limitedDamage^{\Att}_{\Victim}, $\\
\hspace*{10mm} $\psi_1,\delta_2 \to 
t \cdot\limitedDamage^{\Att}_{\Victim}, \underline{b \lap_{\limitedDamage} \Sans}$ & \Ax  \\[2mm]
8 \ $b:{-} t \cdot\limitedDamage^{\Att}_{\Victim}, \Sans \lap_{\limitedDamage} b,\psi_1, \underline{\delta_2} \to$\\
\hspace*{10mm} $ t \cdot\limitedDamage^{\Att}_{\Victim}$ & \rAML:9 \\[2mm]
9 \ $b\lap_{\limitedDamage} \Sans, \underline{b:{-} t \cdot\limitedDamage^{\Att}_{\Victim}}, \Sans \lap_{\limitedDamage} b,\psi_1,\delta_2 \to$\\
\hspace*{10mm} $ t \cdot\limitedDamage^{\Att}_{\Victim}$ & \rNB:10 \\[2mm]
10 \ $b\lap_{\limitedDamage} \Sans ,b:{-} t \cdot\limitedDamage^{\Att}_{\Victim}, \Sans \lap_{\limitedDamage} b,\psi_1,\delta_2 \to$\\
\hspace*{10mm} $ t \cdot\limitedDamage^{\Att}_{\Victim},b: t \cdot\limitedDamage^{\Att}_{\Victim}$ & \rAC %
\\[4mm]\hline}
\caption{Derivation of $\psi_1,\delta_2 \to 
{t \cdot\limitedDamage^{\Att}_{\Victim}}$ (see Subsection~\ref{subsec:conclusions})}\label{fig:der3}
\end{figure}

Let us now introduce the provided derived evidence for this cyber-attack.
The sources provide discordant derived evidence with respect to the motives of the attack, whether the attack had or not economical motives. 
In particular, $\Sans$ states that attacker did not have economical motives ($\econMotives$) to perform the attack as there were not considerable economical losses ($\econLoss$) since the damages caused by the attack to the victim were limited ($\limitedDamage$).
On the other hand, $\R$ states that the attacker was pushed by economical motives because it caused to the victim economical losses and reputation losses ($\repLoss$). 
\begin{displaymath}\small
\renewcommand\arraystretch{1.3}
\begin{array}{l}
\Sans:: (t \cdot\limitedDamage^{\Att}_{\Victim} / {-} t \cdot\econLoss^{\Att}_{\Victim})\\
\Sans:: ({-}t \cdot\econLoss^{\Att}_{\Victim} / {-} t\cdot\econMotives^{\attacker}_{\Att}) \\
\R:: (t \cdot\econLoss^{\Att}_{\Victim}, t\cdot\repLoss^{\Att}_{\Victim} / t\cdot\econMotives^{\attacker}_{\Att})
\end{array}
\end{displaymath}
There are discordant derived evidence also about how the attacker got access to the victim system ($\victimAccess$).  
$\SC$ states that the attacker got access by exploiting the HMI vulnerabilities ($\HMIvuln$), as during the attack the $\BEthree$ malware was used, and this is similar to $\BEtwo$, which usually infects ($\infection$)
the system by exploiting the HMI vulnerabilities. 
$\Sans$ states that the access to the victim system was done using phishing emails ($\phishEmails$) and not through the exploitation of HMI vulnerabilities, because phishing emails were found in the victim system, thus there was a spear-phishing campaign ($\spearPhishing$), the $\BEthree$ malware was found in the phishing emails ($\found$), and the system was infected by $\BEthree$ through the emails. 
$\FE$ states as well that the victim's system was not accessed using HMI vulnerabilities, as the system was not infected by $\BEthree$ through these vulnerabilities, and they were not exploited ($\exploitedVuln$).
\begin{displaymath}\small
\renewcommand\arraystretch{1.3}
\begin{array}{l}
\SC::(t'\cdot\infection^{\BEtwo}_{\HMIvuln},  t\cdot\simil^{\BEthree}_{\BEtwo}/ t'\cdot\infection^{\BEthree}_{\HMIvuln})\\
\SC::  (t'\cdot\infection^{\BEthree}_{\HMIvuln}, t\cdot\usedMalware^{\BEthree}_{\Att}  / t'\cdot\victimAccess^{\Att}_{\HMIvuln})\\

\Sans::  (t'\cdot\spearPhishing^{\phishEmails}_{\Att},  t'\cdot\found^{\BEthree}_{\phishEmails},  t'\cdot\infection^{\BEthree}_{\phishEmails}  / 
\\  \qquad \quad  t'\cdot\victimAccess^{\Att}_{\phishEmails})\\

\Sans::  (t'\cdot\spearPhishing^{\phishEmails}_{\Att},  t'\cdot\found^{\BEthree}_{\phishEmails},  t'\cdot\infection^{\BEthree}_{\phishEmails}  / 
\\ \qquad \quad
{-} t'\cdot\victimAccess^{\Att}_{\HMIvuln})\\

\FE::  ({-} t'\cdot\infection^{\BEthree}_{\HMIvuln}, {-} t'\cdot\exploitedVuln^{\Att}_{\HMIvuln} /{-} t'\cdot\victimAccess^{\Att}_{\HMIvuln})
\end{array}
\end{displaymath}
The  derived evidence provided by the various sources have discordances also for the starting date of the attack.
In particular, $\Sans$ states that the attack started in \emph{Spring 2015}, denoted by $t'$, as in that period there was a spear-phishing campaign, the system was infected with $\BEthree$ through the phishing emails, and a backdoor was created to the victim's system. $\R$ instead states that the attack started in \emph{December 2015}, denoted by $t''$, as the news about the victim's blackout dates to that period. $\FE$ contradicts $\R$, by stating that the attack did not start in $t''$, as the infection occurred in $t'$, and the news about the attack spread in $t''$. $\Sans$ states the same derived evidence as $\FE$, but based on the evidence that the attack started in $t'$ and that $t'$ is different from $t''$.
\begin{displaymath}\small
\renewcommand\arraystretch{1.3}
\begin{array}{l}
\Sans:: (t'\cdot\spearPhishing^{\phishEmails}_{\Att}, t'\cdot\infection^{\BEthree}_{\phishEmails}, t' \cdot\backdoor^{\BEthree}_{\Victim} / \\
\qquad \quad t'\cdot\started^{\Att})\\

\R:: (t''\cdot\news^{\Blackout}_{\Victim} / t''\cdot\started^{\Att})\\

\FE:: (t''\cdot\news^{\Blackout}_{\Victim}, t'\cdot\infection^{\BEthree}_{\phishEmails}/ {-} t''\cdot\started^{\Att})\\ 
 
\Sans:: (t'\cdot\started^{\Att}, t' \not \cong t'' / {-} t''\cdot\started^{\Att})
\end{array}
\end{displaymath}
The last discordance arises about the possible culprit of the attack. $\FE$ states that the group of attackers $\Sandworm$ is a possible culprit, as malware $\BEthree$ was used during $\Att$, which is similar to malware $\BEtwo$ that is used by $\Sandworm$.
$\R$ states that the possible culprit of the attack is not country $\Russia$, as the used malware was $\KillDisk$ and $\KillDisk$ is not used
neither by $\Sandworm$ nor by $\Russia$.
Source $\Blog$ 
states that the possible culprit of the attack is country $\Russia$, as 
there was a \emph{political conflict} in that period between $\Russia$ and the victim $\Victim$. The attack was political, as $\Victim$ did not suffer from economical losses and the source IPs of the attack were geolocated in country $\Russia$.
\begin{displaymath}\small
\renewcommand\arraystretch{1.3}
\begin{array}{l}
\FE::  (t\cdot\usedMalware^{\BEthree}_{\Att}, t\cdot\simil^{\BEthree}_{\BEtwo}, t\cdot\usedBy^{\BEtwo}_{\Sandworm} / 
t\cdot\possibleCulprit^{\Att}_{\Sandworm})\\

\R:: (t \cdot\usedMalware^{\KillDisk}_{\Att}, {-} t\cdot\usedBy^{\KillDisk}_{\Sandworm}, {-} t\cdot\usedBy^{\KillDisk}_{\Russia} / 
\\ \qquad \quad 
{-} t\cdot\possibleCulprit^{\Att}_{\Russia})\\

\Blog:: ({-}t\cdot\econLoss^{\Att}_{\Victim}, 
t\cdot\polConflict^{\Russia}_{\Victim},
t\cdot\geoSourceIP^{\Russia}_{\Att} /
t\cdot\possibleCulprit^{\Att}_{\Russia})
\end{array}
\end{displaymath}

\begin{figure}[t]
\sder{
\hline
\\[-2mm]
1 \ $\psi_1, \underline{\psi_2}, \delta_2 \to \Sans: {-} t \cdot\econLoss^{\Att}_{\Victim}$ & \CR2:2,3 \\[2mm]
2 \ $\psi_1,\delta_2 \to \Sans:{-} t \cdot\econLoss^{\Att}_{\Victim},t \cdot\limitedDamage^{\Att}_{\Victim}$ & (D1) \\[2mm]
3 \ $\underline{\Sans: {-} t \cdot\econLoss^{\Att}_{\Victim},\psi_1}, \delta_2 \to \underline{\Sans: {-} t \cdot\econLoss^{\Att}_{\Victim}}$ &  \Ax
\\[4mm]\hline}
\caption{Derivation of $\psi_1, {\psi_2}, \delta_2 \to \Sans: {-} t \cdot\econLoss^{\Att}_{\Victim}$ (see Subsection~\ref{subsec:conclusions})}\label{fig:der4}
\end{figure}

Let us now introduce the \emph{relations of trust} 
the analyst has with respect to the sources that provided the above evidence. 
In particular, the analyst trusts:
$\Sans$ more than $\R$ about the economical losses of $\Victim$ from $\Att$;
$\ROne$ more than $\In$ and $\SansIn$ more than $\ROne$ for how the $\Att$ was able to infect $\Victim$, ($\infection$);
$\Sans$ more than $\R$ for when the attack started;
$\ROne$
and $\Sans$ more than $\R$ about who has used a particular malware in the past;
and $\FE$ more than $\Blog$ about who is a possible culprit of $\Att$:
\begin{displaymath}\small
\renewcommand\arraystretch{1.3}
\begin{array}{ll}
\R \lap_{\econLoss} \Sans
 \  \ \qquad  \qquad  & \ \ 
\R \lap_{\usedBy} \ROne\\

\In \lap_{\infection} \ROne 
 \  \ \qquad  \qquad  & \ \ 
\R \lap_{\usedBy} \Sans \\

\ROne \lap_{\infection} \SansIn
 \  \ \qquad  \qquad  & \ \ 
\Blog \lap_{\possibleCulprit} \FE\\

\R \lap_{\started} \Sans
 \  \ \qquad  \qquad  & \ \ 
\end{array}
\end{displaymath}
\begin{figure}[t]
\sder{
\hline
\\[-2mm]
1 \ $\psi_1, \underline{\psi_2}, \delta_1, \delta_2 \to {-} t \cdot\econLoss^{\Att}_{\Victim}$ & \CR2:2,3 \\[1mm]
2 \ $\psi_1,\delta_1,\delta_2 \to {-} t \cdot\econLoss^{\Att}_{\Victim},t \cdot\limitedDamage^{\Att}_{\Victim}$ & (D1) \\[1mm]
3 \ $\Sans: {-} t \cdot\econLoss^{\Att}_{\Victim},\psi_1,\delta_1,\delta_2 \to \underline{{-} t \cdot\econLoss^{\Att}_{\Victim}}$ &  \rBKR:4,5,6\\[1mm]
4 \ $b:t \cdot\econLoss^{\Att}_{\Victim}, \underline{b: \lnecb t \cdot\econLoss^{\Att}_{\Victim}}, \Sans: {-} t \cdot\econLoss^{\Att}_{\Victim},\psi_1,\delta_1,\delta_2 \to$ & \rBCL:7,8 \\[1mm]
5 \ $\underline{\Sans: {-} t \cdot\econLoss^{\Att}_{\Victim}}, \psi_1, \delta_1, \delta_2 \to {-} t \cdot\econLoss^{\Att}_{\Victim}, \underline{\Sans: {-} t \cdot\econLoss^{\Att}_{\Victim}}$ & \Ax\\[1mm]
6 \ $\Sans: {-} t \cdot\econLoss^{\Att}_{\Victim},\psi_1,\delta_1,\delta_2 \to {-} t \cdot\econLoss^{\Att}_{\Victim},$\\
\hspace*{10mm} $\underline{\Sans: \lnecb ({-} t \cdot\econLoss^{\Att}_{\Victim})}$ & \rBCR:9\\[1mm]
7 \ $b:t \cdot\econLoss^{\Att}_{\Victim}, b: \lnecb t \cdot\econLoss^{\Att}_{\Victim}, \Sans: {-} t \cdot\econLoss^{\Att}_{\Victim}, \psi_1, \underline{\delta_1}, \delta_2 \to$ \\
\hspace*{10mm} $b \lap_{\econLoss} \Sans$ & \rAML:12 \\[1mm]
8 \ $b:t \cdot\econLoss^{\Att}_{\Victim}, b: \lnecb t \cdot\econLoss^{\Att}_{\Victim}, \underline{\Sans: {-} t \cdot\econLoss^{\Att}_{\Victim}}, \psi_1, \delta_1, \delta_2 \to  $ \\
\hspace*{10mm} $\underline{\Sans: {-} t \cdot\econLoss^{\Att}_{\Victim}}$ & \Ax \\[1mm]
9 \ $b: t \cdot\econLoss^{\Att}_{\Victim}, \Sans \lap_{\econLoss} b, \Sans: {-} t \cdot\econLoss^{\Att}_{\Victim}, \psi_1, \underline{\delta_1}, \delta_2 \to$\\
\hspace*{10mm} ${-} t \cdot\econLoss^{\Att}_{\Victim}$ & \rAML:10\\[1mm]
10 \ $b \lap_{\econLoss} \Sans, b: t \cdot\econLoss^{\Att}_{\Victim}, \Sans \lap_{\econLoss} b, \underline{\Sans: {-} t \cdot\econLoss^{\Att}_{\Victim}},$\\
\hspace*{10mm} $\psi_1,\delta_1,\delta_2 \to {-} t \cdot\econLoss^{\Att}_{\Victim}$ & \rNB:11\\[1mm]
11 \ $b \lap_{\econLoss} \Sans, b: t \cdot\econLoss^{\Att}_{\Victim}, \Sans \lap_{\econLoss} b, \Sans: {-} t \cdot\econLoss^{\Att}_{\Victim},$\\
\hspace*{10mm} $\psi_1,\delta_1,\delta_2 \to {-} t \cdot\econLoss^{\Att}_{\Victim}, \Sans: t \cdot\econLoss^{\Att}_{\Victim}$ & \rAC\\[1mm]
12 \ $\underline{b \lap_{\econLoss} \Sans}, b:t \cdot\econLoss^{\Att}_{\Victim}, b: \lnecb t \cdot\econLoss^{\Att}_{\Victim},$ \\
\hspace*{10mm} $\Sans: {-} t \cdot\econLoss^{\Att}_{\Victim},\psi_1,\delta_1,\delta_2 \to \underline{b \lap_{\econLoss} \Sans}$ & \Ax
\\[4mm]\hline}
\caption{Derivation of $\psi_1, {\psi_2}, \delta_1, \delta_2 \to {-} t \cdot\econLoss^{\Att}_{\Victim}$ (see Subsection~\ref{subsec:conclusions})}\label{fig:der5}
\end{figure}
Moreover:
$\Sans$ is the most trusted agent about how the attacker accessed the victim's system, the economical losses, the economical motives and the limited damages; 
$\BBC$ is the most trusted about identifying political conflict between entities;
$\DHS$ is the most trusted about the similarities between malware and used malware in attacks;
$\FE$ is the most trusted about geolocating the source IPs of an attack;
$\FE$ is the most trusted regarding the statement that $\BEthree$ and $\BEtwo$ are similar between each other.

\begin{displaymath}\small
\renewcommand\arraystretch{1.3}
\begin{array}{ll}
\forall x. \, x \lap_{\victimAccess} \Sans
 \  \ \qquad  \qquad  & \ \ 
\forall x. \, x \lap_{\simil} \DHS\\

\forall x. \, x \lap_{\econLoss} \Sans
 \  \ \qquad  \qquad  & \ \ 
 \forall x. \, x \lap_{\usedBy} \DHS\\

\forall x. \, x \lap_{\econMotives} \Sans
 \  \ \qquad  \qquad  & \ \ 
\forall x. \, x \lap_{\usedMalware} \Sans\\

\forall x. \, x \lap_{\limitedDamage} \Sans
 \  \ \qquad  \qquad  & \ \ 
\forall x. \, x \lap_{\geoSourceIP} \FE\\

\forall x. \, x \lap_{\polConflict} \BBC
 \  \ \qquad  \qquad  & \ \ 
\FE: \lnecb \simil^{\BEthree}_{\BEtwo}

\end{array}
\end{displaymath}

\subsection{The Conclusions of the Analyst}\label{subsec:conclusions}
In the following, we show how the analyst can use the Time-Stamped Claim Logic to draw some interesting conclusions from the collected evidence.

Consider the following assertions:
\begin{displaymath}
\renewcommand\arraystretch{1.3}
\begin{array}{l}
\psi_1= \Sans: t\cdot\limitedDamage^{\Att}_{\Victim}\\
\psi_2 = \Sans:: (t \cdot\limitedDamage^{\Att}_{\Victim} / {-} t \cdot\econLoss^{\Att}_{\Victim}) \\
\psi_3 = \Sans:: ({-}t \cdot\econLoss^{\Att}_{\Victim} / {-} t\cdot\econMotives^{\attacker}_{\Att})\\
\delta_1 = \forall x. \, x \lap_{\econLoss} \Sans \\
\delta_2 = \forall x. \, x \lap_{\limitedDamage} \Sans\\
\delta_3 = \forall x. \, x \lap_{\econMotives} \Sans
\end{array}.
\end{displaymath}
We start by giving the derivation in Figure~\ref{fig:der3} of 
\begin{displaymath}
t \cdot\limitedDamage^{\Att}_{\Victim}
\end{displaymath}
from the assumptions $\psi_1$ and $\delta_2$. 
We can then show that
\begin{displaymath}
\Sans: {-} t \cdot\econLoss^{\Att}_{\Victim}
\end{displaymath}
is derived from $\psi_1$, $\psi_2$ and $\delta_2$ as can be seen in Figure~\ref{fig:der4}. 
The derivation in Figure~\ref{fig:der5} infers 
\begin{displaymath}
{-} t \cdot\econLoss^{\Att}_{\Victim}
\end{displaymath}
from $\psi_1$, $\psi_2$, $\delta_1$ and $\delta_2$. 
From $\psi_1$, $\psi_2$, $\psi_3$, $\delta_1$ and $\delta_2$, we can also derive 
$$
\Sans:{-} t\cdot\econMotives^{\attacker}_{\Att}
$$
as can be seen in Figure~\ref{fig:der6}.
\begin{figure}[t]
\sder{
\hline
\\[-2mm]
1 \ $\psi_1, \psi_2, \underline{\psi_3}, \delta_1, \delta_2 \to \Sans:{-} t\cdot\econMotives^{\attacker}_{\Att}$ & \CR2:2,3 \\[2mm]
2 \ $\psi_1,\psi_2,\delta_1,\delta_2 \to {-} t \cdot\econLoss^{\Att}_{\Victim}, \Sans:{-} t\cdot\econMotives^{\attacker}_{\Att}$ & (D2) \\[2mm]
3 \ $\underline{\Sans:{-} t\cdot\econMotives^{\attacker}_{\Att}}, \psi_1, \psi_2, \delta_1, \delta_2 \to \underline{\Sans:{-} t\cdot\econMotives^{\attacker}_{\Att}}$ & \Ax
\\[4mm]\hline}
\caption{Derivation of $\psi_1, \psi_2, {\psi_3}, \delta_1, \delta_2 \to \Sans:{-} t\cdot\econMotives^{\attacker}_{\Att}$ (Subsection~\ref{subsec:conclusions})}\label{fig:der6}
\end{figure}
From $\psi_1$, $\psi_2$, $\psi_3$, $\delta_1$,  $\delta_2$ and $\delta_3$ we can derive  
$$
{-} t\cdot\econMotives^{\attacker}_{\Att}
$$
as can be seen in Figure~\ref{fig:der7}.
\begin{figure}[h]
{\small
\sder{
\hline
\\[-2mm]
1 \ $\psi_1, \underline{\psi_2}, \psi_3, \delta_1, \delta_2, \delta_3 \to{-} t\cdot\econMotives^{\attacker}_{\Att}$ & \CR2:2,3 \\[2mm]
2 \ $\psi_1,\psi_2,\delta_1,\delta_2,\delta_3 \to {-} t\cdot\econMotives^{\attacker}_{\Att},{-}t \cdot\econLoss^{\Att}_{\Victim}$ & (D2) \\[2mm]
3 \ $\Sans: {-} t\cdot\econMotives^{\attacker}_{\Att},\psi_1,\psi_2,\delta_1,\delta_2,\delta_3 \to \underline{{-} t\cdot\econMotives^{\attacker}_{\Att}}$ &  \rBKR:4,5,6\\[2mm]
4 \ $b: t\cdot\econMotives^{\attacker}_{\Att}, \underline{b: \lnecb t\cdot\econMotives^{\attacker}_{\Att}},$\\
\hspace*{10mm} $\Sans: {-} t\cdot\econMotives^{\attacker}_{\Att},\psi_1,\psi_2,\delta_1,\delta_2,\delta_3 \to$ & \rBCL:7,8 \\[2mm]
5 \ $\underline{\Sans: {-} t\cdot\econMotives^{\attacker}_{\Att}}, \psi_1, \psi_2, \delta_1, \delta_2, \delta_3 \to $\\
\hspace*{10mm} ${-} t\cdot\econMotives^{\attacker}_{\Att}, \underline{\Sans: {-} t\cdot\econMotives^{\attacker}_{\Att}}$ & \Ax\\[2mm]
6 \ $\Sans: {-} t\cdot\econMotives^{\attacker}_{\Att},\psi_1,\psi_2,\delta_1,\delta_2,\delta_3 \to {-} t\cdot\econMotives^{\attacker}_{\Att},$\\
\hspace*{10mm} $\underline{\Sans: \lnecb ({-} t\cdot\econMotives^{\attacker}_{\Att})}$ & \rBCR:9\\[2mm]
7 \ $b:t\cdot\econMotives^{\attacker}_{\Att}, b: \lnecb  t\cdot\econMotives^{\attacker}_{\Att},$\\
\hspace*{10mm} $\Sans: {-}t\cdot\econMotives^{\attacker}_{\Att}, \psi_1, \psi_2, \delta_1, \delta_2, \underline{\delta_3} \to b \lap_{\econMotives} \Sans$ & \rAML:12 \\[2mm]
8 \ $b: t\cdot\econMotives^{\attacker}_{\Att}, b: \lnecb  t\cdot\econMotives^{\attacker}_{\Att},$\\
\hspace*{10mm} $\underline{\Sans: {-} t\cdot\econMotives^{\attacker}_{\Att}}, \psi_1, \psi_2, \delta_1, \delta_2, \delta_3 \to  $ \\
\hspace*{10mm} $\underline{\Sans: {-} t\cdot\econMotives^{\attacker}_{\Att}}$ & \Ax \\[2mm]
9 \ $b: t\cdot\econMotives^{\attacker}_{\Att}, \Sans \lap_{\econMotives} b, \Sans: {-} t\cdot\econMotives^{\attacker}_{\Att},$\\
\hspace*{10mm} $\psi_1, \psi_2, \delta_1, \delta_2, \underline{\delta_3} \to {-} t\cdot\econMotives^{\attacker}_{\Att}$ & \rAML:10\\[2mm]
10 \ $b \lap_{\econMotives} \Sans, b: t\cdot\econMotives^{\attacker}_{\Att}, \Sans \lap_{\econMotives} b, $\\
\hspace*{10mm} $\underline{\Sans: {-} t\cdot\econMotives^{\attacker}_{\Att}}, \psi_1, \psi_2, \delta_1, \delta_2, \delta_3 \to {-} t\cdot\econMotives^{\attacker}_{\Att}$ & \rNB:11\\[2mm]
11 \ $b \lap_{\econMotives} \Sans, b: t\cdot\econMotives^{\attacker}_{\Att}, \Sans \lap_{\econMotives} b, $\\
\hspace*{10mm} $\Sans: {-} t\cdot\econMotives^{\attacker}_{\Att},\psi_1,\psi_2,\delta_1,\delta_2,\delta_3 \to$\\
\hspace*{10mm} $ {-} t\cdot\econMotives^{\attacker}_{\Att},\Sans:  t\cdot\econMotives^{\attacker}_{\Att}$ & \rAC\\[2mm]
12 \ $\underline{b \lap_{\econMotives} \Sans}, b:t\cdot\econMotives^{\attacker}_{\Att}, b: \lnecb  t\cdot\econMotives^{\attacker}_{\Att},$\\
\hspace*{10mm} $\Sans: {-}t\cdot\econMotives^{\attacker}_{\Att},\psi_1,\psi_2,\delta_1,\delta_2,\delta_3 \to$ \\
\hspace*{10mm} $\underline{b \lap_{\econMotives} \Sans}$ & \Ax
\\[4mm]\hline}}
\caption{Derivation of $\psi_1, {\psi_2}, \psi_3, \delta_1, \delta_2, \delta_3 \to{-} t\cdot\econMotives^{\attacker}_{\Att}$ }\label{fig:der7}
\end{figure}
From
\begin{displaymath}\small
\renewcommand{\arraystretch}{1.3}
\begin{array}{l}
\psi_4 = \FE::  (t\cdot\usedMalware^{\BEthree}_{\Att}, t\cdot\simil^{\BEthree}_{\BEtwo}, t\cdot\usedBy^{\BEtwo}_{\Sandworm} / 
t\cdot\possibleCulprit^{\Att}_{\Sandworm}) \\

\psi_5 = \Sans: t\cdot\usedMalware^{\BEthree}_{\Att} \\

\psi_6 = \DHS: t\cdot\simil^{\BEthree}_{\BEtwo}\\

\psi_7 = \DHS: t\cdot\usedBy^{\BEtwo}_{\Sandworm}\\

\delta_4 = \forall x. \, x \lap_{\simil} \DHS\\
\delta_5 = \forall x. \, x \lap_{\usedBy} \DHS\\
\delta_6 = \forall x. \, x \lap_{\usedMalware} \Sans
\end{array}
\end{displaymath}
we can conclude
$$
\FE: t\cdot\possibleCulprit^{\Att}_{\Sandworm}
$$
as can be seen in Figure~\ref{fig:der8}.
\begin{figure}[h]
{\small
\sder{
\hline
\\[-2mm]
1 \ $\underline{\psi_4}, \psi_5, \psi_6, \psi_7, \delta_4, \delta_5, \delta_6 \to \FE: t\cdot\possibleCulprit^{\Att}_{\Sandworm}$ & \CR2:2--5 \\[2mm]
2 \ $\psi_5, \psi_6, \psi_7, \delta_4, \delta_5, \delta_6 \to \FE: t\cdot\possibleCulprit^{\Att}_{\Sandworm}, \underline{t\cdot\usedMalware^{\BEthree}_{\Att}}$ & \rBKR:6--8\\[2mm]
3 \ $\psi_5, \psi_6, \psi_7, \delta_4, \delta_5, \delta_6 \to \FE: t\cdot\possibleCulprit^{\Att}_{\Sandworm}, \underline{t\cdot\simil^{\BEthree}_{\BEtwo}}$ & similar to 2 \\[2mm]
4 \ $\psi_5, \psi_6, \psi_7, \delta_4, \delta_5, \delta_6 \to \FE: t\cdot\possibleCulprit^{\Att}_{\Sandworm}, \underline{t\cdot\usedBy^{\BEtwo}_{\Sandworm}}$ & similar to 2 \\[2mm]
5 \ $\underline{\FE: t\cdot\possibleCulprit^{\Att}_{\Sandworm}}, \psi_5, \psi_6, \psi_7, \delta_4, \delta_5, \delta_6 \to \underline{\FE: t\cdot\possibleCulprit^{\Att}_{\Sandworm}}$ & \Ax \\[2mm]
6 \ $b: {-} (t\cdot\usedMalware^{\BEthree}_{\Att}), \underline{b: \lnecb {-}(t\cdot\usedMalware^{\BEthree}_{\Att})}, \psi_5, \psi_6, \psi_7, \delta_4, \delta_5, \delta_6 \to$ \\
\qquad $\FE: t\cdot\possibleCulprit^{\Att}_{\Sandworm}$ & 
\rBCL:9,10
\\[2mm]
7 \ $\underline{\psi_5}, \psi_6, \psi_7, \delta_4, \delta_5, \delta_6 \to \FE: t\cdot\possibleCulprit^{\Att}_{\Sandworm}, t\cdot\usedMalware^{\BEthree}_{\Att},$ \\
\qquad $\underline{\Sans: t\cdot\usedMalware^{\BEthree}_{\Att}}$ & \Ax \\[2mm]
8 \ $\psi_5, \psi_6, \psi_7, \delta_4, \delta_5, \delta_6 \to \FE: t\cdot\possibleCulprit^{\Att}_{\Sandworm}, t\cdot\usedMalware^{\BEthree}_{\Att}$ \\
\qquad $\underline{\Sans: \lnecb t\cdot\usedMalware^{\BEthree}_{\Att}}$ & \rBCR:11 \\[2mm]
9 \ $b: {-}(t\cdot\usedMalware^{\BEthree}_{\Att}), b: \lnecb {-}(t\cdot\usedMalware^{\BEthree}_{\Att}), \psi_5, \psi_6, \psi_7, \delta_4, \delta_5, \underline{\delta_6} \to$ \\
\qquad $\FE: t\cdot\possibleCulprit^{\Att}_{\Sandworm}, b \lap_{\usedMalware} \Sans$ & \rAML:14 \\[2mm]
10 \ $b: {-}(t\cdot\usedMalware^{\BEthree}_{\Att}), b: \lnecb {-}(t\cdot\usedMalware^{\BEthree}_{\Att}), \underline{\psi_5}, \psi_6, \psi_7, \delta_4, \delta_5, \delta_6 \to$ \\
\qquad $\FE: t\cdot\possibleCulprit^{\Att}_{\Sandworm}, \underline{\Sans: t\cdot\usedMalware^{\BEthree}_{\Att}}$ & \Ax \\[2mm]
11 \ $\underline{c: {-}(t\cdot\usedMalware^{\BEthree}_{\Att})}, \Sans \lap_{\usedMalware} c, \psi_5, \psi_6, \psi_7, \delta_4, \delta_5, \delta_6 \to$ \\
\qquad $\FE: t\cdot\possibleCulprit^{\Att}_{\Sandworm}, t\cdot\usedMalware^{\BEthree}_{\Att}$ & \rNB:12 \\[2mm]
12 \ $c: {-}(t\cdot\usedMalware^{\BEthree}_{\Att}), \Sans \lap_{\usedMalware} c, \psi_5, \psi_6, \psi_7, \delta_4, \delta_5, \underline{\delta_6} \to$ \\
\qquad $\FE: t\cdot\possibleCulprit^{\Att}_{\Sandworm}, t\cdot\usedMalware^{\BEthree}_{\Att}, c: t\cdot\usedMalware^{\BEthree}_{\Att}$ & \rAML:13 \\[2mm]
13 \ $\underline{c \lap_{\usedMalware} \Sans}, c: {-}(t\cdot\usedMalware^{\BEthree}_{\Att}), \underline{\Sans \lap_{\usedMalware} c}, $ \\
\qquad $\underline{\psi_5}, \psi_6, \psi_7, \delta_4, \delta_5, \delta_6 \to\FE: t\cdot\possibleCulprit^{\Att}_{\Sandworm}, t\cdot\usedMalware^{\BEthree}_{\Att},$ \\
\qquad $\underline{c: t\cdot\usedMalware^{\BEthree}_{\Att}}$ & \rAC \\[2mm]
14 \ $\underline{b \lap_{\usedMalware} \Sans}, b: {-} t\cdot\usedMalware^{\BEthree}_{\Att}, b: \lnecb {-} t\cdot\usedMalware^{\BEthree}_{\Att}, $ \\
\qquad $\psi_5, \psi_6, \psi_7, \delta_4, \delta_5, \delta_6 \to \FE: t\cdot\possibleCulprit^{\Att}_{\Sandworm}, \underline{b \lap_{\usedMalware} \Sans}$ & \Ax
\\[4mm]\hline}}
\caption{Derivation of ${\psi_4}, \psi_5, \psi_6, \psi_7, \delta_4, \delta_5, \delta_6 \to \FE: t\cdot\possibleCulprit^{\Att}_{\Sandworm}$}\label{fig:der8}
\end{figure}
\noindent
Finally, it is not difficult to conclude
$$
\Blog: t\cdot\possibleCulprit^{\Att}_{\Russia} 
$$
from 
\begin{displaymath}\small
\renewcommand\arraystretch{1.3}
\begin{array}{l}
\Blog:: ({-}t\cdot\econLoss^{\Att}_{\Victim}, 
t\cdot\polConflict^{\Russia}_{\Victim},
t\cdot\geoSourceIP^{\Russia}_{\Att} /
t\cdot\possibleCulprit^{\Att}_{\Russia})  \\

\FE : t\cdot\geoSourceIP^{\Russia}_{\Att}\\

\BBC: t\cdot\polConflict^{\Russia}_{\Victim} \\

\Sans:: (t \cdot\limitedDamage^{\Att}_{\Victim} / {-} t \cdot\econLoss^{\Att}_{\Victim})\\

\forall x. \, x \lap_{\econLoss} \Sans\\

\forall x. \, x \lap_{\geoSourceIP} \FE\\

\forall x. \, x \lap_{\polConflict} \BBC
\end{array}
\end{displaymath}
and some of the above results. However, thanks to the trust relation 
\begin{displaymath}
\Blog \lap_{\possibleCulprit} \FE\,,
\end{displaymath}
which says that the analyst trusts $\FE$ more than $\Blog$ for what concerns the possible culprit of the attack, the analyst's final result will discard $\Blog$'s statement, and the analyst will conclude that the possible culprit of the attack is $\Sandworm$.

\subsection{Summary}

In this case study, one of the forensics analyst's conclusions is that a \emph{possible culprit} of the attack is the \emph{group of hackers} denoted by $\Sandworm$. This conclusion was derived using the analyst's relation of trust,
where the analyst trusts the security company $\FE$ more than $\Blog$. Given the set of evidence collected, our Time-Stamped Claim Logic allows the analyst to filter the conflicting evidence by using his relations of trust. This reduces the number of evidence that the analyst needs to consider during the analysis phase. 
Furthermore, the set of evidence that the analyst obtains in this way is consistent, as all the conflicting pieces of evidence are removed (assuming that the analyst has provided all the required trust relations). 

Our Time-Stamped Claim Logic provides to the forensics analyst with a sound and complete means to \emph{reduce the size} of the collected set of evidence and \emph{remove inconsistencies}, i.e., the logic ensures that the result is consistent with respect to the trust relations that the analyst considers to hold. 
In other words, the logic
uses the evidence and the derivation of evidence of trusted sources together with the analyst's relations of trust to remove the inconstancies and provide to the analyst only evidence that is trustworthy (from his point of view). Furthermore, the logic is able to provide to the analyst \emph{new evidence} derived using the given set of evidence. The derived evidence can be general evidence that is always true or statements that sources have made and that the analyst trusts.

\section{Related Work}
\label{sec:Comparison}

In this section, we discuss the logics and approaches that are most closely related to our Time-Stamped Claim Logic.

We begin by pointing out that there are a number of recent works on reasoning about claims from different logical points of view, namely \emph{Justification Logics} (see~\cite{art:08,fit:08,buc:11}), which capitalizes on the \emph{Logic of Proofs} (see~\cite{art:01}), \emph{Modal Evidence Logic} (see~\cite{ben:11a,ben:12,bal:14}) and a logical account of \emph{Formal Argumentation} (see~\cite{gab:09}). A recent work presents a paraconsistent logic able to reason about preservation of evidence and of truth (see~\cite{wac:ar:17}).

In \emph{Justification Logics}~\cite{art:08,fit:08,buc:11} propositional formulas are labeled with justification terms built from variables and constants using several operations. 
The semantics is provided by Kripke structures enriched with an evidence function that associates a set of worlds to  each propositional formula and justification term. Justifications in this logic are seen as the strongest form of evidence providing a direct proof of the truth of an assertion. 

More specifically, Justification Logics
are enrichments of classical propositional logic with assertions of the form
$$
s:F\,,
$$
where $s$ is a term and $F$ is a propositional formula, meaning that $s$ is a justification for $F$. The operations allowed on the labels express the possible interactions that justifications may have. The logic we propose differs significantly from Justification Logics since our main goal is to address distributed time-stamped claims instead of justifications. Reflecting this different goal, the labels in the Time-Stamped Claim Logic are an agent and/or a time-stamp. Furthermore, our logic includes explicit constructions for expressing trust relations between agents with respect to relevant evidence. On the other hand, Justification Logics may allow nesting of labels which our logic currently does not support. Both Justification Logics and Time-Stamped Claim Logic allow reasoning about labeled formulas and formulas without labels. However, the relationship between these formulas is made clear at the meta-level in Justification Logics, whereas in the calculus of Time-Stamped Claim Logic this relationship is internalized.
 
The \emph{Modal Evidence Logic}~\cite{ben:11a,ben:12,bal:14} provides a way to model epistemic agents when dealing with evidence from different sources. It is a multi-modal logic that comprises three modal like operators: one for representing that an agent has evidence for a formula, another for stating that an agent believes a formula to be true and a last one asserting that a formula is true in every world.
Neighborhood semantics is adopted. 

An important difference between Modal Evidence Logic and our Time-Stamped Claim Logic is that Modal Evidence Logic is not a labeled logic. Consequently, no explicit reasoning about agents is possible in Modal Evidence Logic contrarily to what happens in our logic. Moreover, in Modal Evidence Logic there is no specific mechanism for dealing with time-stamped claims. 

\emph{Formal argumentation} has been investigated since the early 1990s but it was the landmark work in~\cite{dun:95} that started an abstract perspective on argumentation. A logic foundation of formal argumentation was proposed in~\cite{gab:09} using classical logic as well as modal logic to characterize different extensions for argumentation. 
The main difference between that work 
and Time-Stamped Claim Logic is that the former is related to arguments and the latter is concerned with agent claims. So, naturally, there is an attack relation between arguments in Formal Argumentation, whereas there is a trustful relationship between agents with respect to a given subject in Time-Stamped Claim Logic. In both Formal Argumentation and Time-Stamped Claim Logic, there are three possibilities for the valuation of both an argument and a claim. Hence, both have three-valued interpretations. In Formal Argumentation, an argument can be explicitly accepted (truth value {\em in}), explicitly rejected (truth value {\em out}) or neither explicitly accepted nor explicitly rejected (truth value {\em undec}). In our logic, either an agent $d$ claims that $p$ occurred at time $n$  (truth value $1$), or an agent $d$ claims that $p$ did not occur at time $n$  (truth value $0$) or an agent $d$ does not have an opinion about the occurrence  of $p$ at time $n$  (truth value $\frac12$).
 
\emph{Digital forensics} techniques suffer from quantity and complexity problems~\cite{Beebe} due to the necessity to handle an enormous amount of pieces of evidence and to the low format of the collected evidence. In this work, we proposed a solution that deals with the quantity problem: Time-Stamped Claim Logic reduces drastically the number of given evidence using the relations of trust of the analyst. 
Our logic allows the forensics team to decrease the resources spent on the analysis, which is usually a human-based process, and to arrive to a swift conclusion. Filtering the evidence reduces the chances of human error as the analyst needs to deal with a considerably lower number of pieces of evidence during the analysis phase. 

The gathered pieces of evidence can be in conflict with each other because they are collected from difference sources and because of the anti-forensics~\cite{goutam15,wheeler03} and deception techniques~\cite{Almeshekah2014,Almeshekah2016} used by the attackers to hide their traces. Digital forensics techniques are able to deal with conflicting information~\cite{aziz,fontani,hu,KarafiliWKL18}, but without taking into account the analyst's relations of trust.  
Our logic solves the conflicts by using in the derivations only trusted pieces of evidence, thus providing as conclusions only consistent sets of pieces of evidence. 

Filtering the evidence of the attack using the Time-Stamped Claim Logic helps the analyst to swiftly identify who has performed the attack with a higher level of confidence. In particular, the results of the filtering process can be used by the analyst during the \emph{attribution phase}~\cite{KarafiliWKL18}: the analyst can then deal only with trusted evidence and has a higher level of confidence for the given result. Using this filtering process would reduce the resources and time spent during the attribution phase, which is a crucial aspect of digital forensics investigations.

In~\cite{kar:vig:18}, the authors have proposed another formal solution for filtering the evidence of cyber-attacks. They introduced the \emph{Evidence Logic (EL)}, which is based on an enriched Linear Temporal Logic and represents the evidence (simple and derived), the sources, the reasonings, the instants of time, and the analyst's relations of trust with respect to the sources and the used reasoning. EL and Time-Stamped Claim Logic obviously share some similarities, starting with the objective of providing security analysts with the formal means to reduce the size of the collected set of evidence and remove inconsistencies. There are also some fundamental differences, though. Most notably, EL is based on a monotonic-reasoning procedure that rewrites the pieces of evidence with the use of tableau rules, but the rewriting system and algorithm that govern the inferences lack a completeness proof. There are also differences in the expressive power of the two logics. In particular, EL has a restrictive constraint, where an evidence is either given or derived, whereas the Time-Stamped Claim Logic is more expressive, as it permits one to represent the same evidence as given and derived (for instance, in the case study in Section~\ref{sec:case-study}, the evidence $\econLoss$ is both given and derived). The Time-Stamped Claim Logic permits one to better represent real evidence that can be both given and derived by other evidence. Moreover, the Time-Stamped Claim Logic uses in its reasonings only trusted evidence, whereas EL has a more \emph{credulous} nature, where it considers in its reasoning also the simple evidence used by the sources to derive new evidence. 

There are differences also at the level of trust relations. The trust relations in EL are binary relations between two different sources, whereas the Time-Stamped Claim Logic has also \emph{universal} trust relations, where a certain source is the most trusted for a particular evidence ($\forall x. \, x \lap_p a$ and $a: \lnecb \phi$). 
EL is able to represent the relation $\forall x. \, x \lap_p a$ by introducing all the needed binary relations of trust (if there are a finite number of sources), but in case new evidence is added, it needs to add new trust relations as well. However, EL is not able to represent the $a: \lnecb \phi$ relation.  

EL represents the \emph{reasonings} followed by the source to derive an evidence, as well as the trust relations between reasonings. The reasonings can be applied by different sources to derive the same conclusion, using the same premises. The Time-Stamped Claim Logic does not represent the reasonings or their trust relations, instead, it uses the trust relations between sources. 
The use of the reasonings and their trust relations in EL permits it to be more generic, whereas the use of trust relations between sources for the derived evidence gives a higher level of specificity to the Time-Stamped Claim Logic. 

Finally, it is interesting to note that, for all these reasons, EL is not able to handle the case study that we considered in this paper in the same way as the Time-Stamped Claim Logic. In order to allow an analyst to derive the same conclusions we obtained here, EL would need to undergo some changes which would come at the cost of losing expressivity.

\section{Concluding Remarks and Future Work}\label{sec:concrems}

We proposed a logic for reasoning about time-stamped claims by presenting a Gentzen calculus and a semantics, and proved soundness and completeness. We also showed how it can be applied concretely to a case study in cyber-security attribution. There are several directions for future work. The first is to provide automated support for our calculus (e.g., by implementing it in a logical framework such as Isabelle or as a tableaux system) in order to avoid having to carry out the derivations by hand. We also intend to investigate whether the logic is decidable~\cite{rss:15} and whether the logic has  Craig interpolation~\cite{srs:13}. Moreover, it would be interesting to generalize the language of the logic in order to address other general claims as well as to enrich the logic with dynamic features~\cite{dit:08}. It also seems natural and worthwhile to
incorporate probabilistic primitives and reasoning in the logic~\cite{lur:18,kok:16,srs:17,hal:06}. Finally, it seems worthwhile to investigate the combination of this logic with logics addressing different aspects in order to obtain a logic with a broader applicability~\cite{gab:09a,ras:02}.

\section*{Acknowledgments}

João Rasga and Cristina Sernadas  deeply acknowledge  the National Funding from FCT (Fundação para a Ciência e a Tecnologia) under the project UID/MAT/04561/2019 granted to CMAFcIO (Centro de Matemática, Aplicações Fundamentais e Investigação Operacional) of Universidade de Lisboa. 
Erisa Karafili was supported by the European Union's H2020 research and innovation programme under the Marie Curie grant agreement No.~746667.

\bibliographystyle{plain}
\bibliography{evl02}

\appendix
\section{Some Proofs of Sections~\ref{sec:soundness} and~\ref{sec:completeness}}

\begin{prop}[a.k.a.~Proposition~\ref{prop:rulessound}]
The rules of the sequent calculus are sound.
\end{prop}
\begin{proof} 
It is quite straightforward to see that the rules $\Ax$, $\Cut$, $\ER$, $\ES$ and $\ET$ are sound.
\\[2mm]
$(\EC)$ Assume that $I \rho \sat t_1 \cong t_2$, $I\rho \sat [\psi]^{t_1}_{t_2}$ and $I \rho \sat \Gamma$. 
When $\var_T(\psi)$ is not $t_1$ then $\psi$ is  $[\psi]^{t_1}_{t_2}$ and so the thesis follows immediately. 
Otherwise, if $\psi$ is $a: \phi$, then the result is immediate by definition of $I$.
\\[2mm]
$(\rAT)$ Assume that $I \rho \sat  a_1 \lap_{p}a_2$,  $I \rho \sat a_2 \lap_{p} a_3$ and $I \rho \sat \Gamma$. The thesis follows immediately by the definition of interpretation structure since $\lap_p^I$ is transitive.
\\[2mm]
$(\rAI)$ Observe that $I \rho \sat a \lap_{p}a$ since $\lap_{p}^I$ is reflexive. 
 \\[2mm]
$(\rAC)$ Assume that $I \rho \sat a_1 \lap_p a_2$, $I \rho \sat a_2 \lap_p a_1$, $I\rho \sat [a:\phi]^{a_1}_{a_2}$ and $I \rho \sat \Gamma$. When $a$ is not $a_1$ then $[a:\phi]^{a_1}_{a_2}$ is  $a:\phi$ and so the thesis follows immediately. Otherwise, the result is obtained by the definition of $I$. \\[2mm]
$(\rAMR)$ Assume that $I \rho' \sat \Gamma \to \Delta, b \lap_p a$ for every $\rho'$ such that  $\rho' \equiv^A_b \rho$ and $I \rho \sat \Gamma$. Suppose that there is no $\delta \in \Delta$ such that $I \rho \sat \delta$. Then $I \rho' \sat b \lap_p a$ for every $\rho'$ such that  $\rho' \equiv^A_b \rho$. Hence, $d \lap_p^{I} \rho_a(a)$ for each $d \in D_A$, and so $I \rho \sat \forall x. \, x \lap_p a$. \\[2mm]
$(\rAML)$ Assume that $I \rho \sat a' \lap_p a, \forall x. \, x \lap_p a, \Gamma \to \Delta$, $I \rho \sat \forall x. \, x \lap_p a$ and $I \rho \sat \Gamma$. Then $d \lap_p^I \rho_A(a)$ for every $d \in D_A$. Hence, in particular, $\rho_A(a') \lap_p^I \rho_A(a)$ and so $I \rho \sat a' \lap_p a$.
Therefore, by the hypothesis, there is $\delta \in \Delta$ such that $I \rho \sat \delta$.\\[2mm]
$(\rNB)$ Assume that $I \rho \sat a: {-}\phi,\Gamma \to \Delta,a:\phi$,  $I \rho \sat \Gamma$ and  $I \rho \sat a: {-} \phi$. 
We just consider the case that $\phi$ is $t\cdot p$. Thus $V(p,\rho_A(a),\rho_T(t))=0$. Therefore, $I \rho  \not \sat  a:\phi$ and so, by the hypothesis, 
there is $\delta \in \Delta$ such that $I \rho \sat \delta$. \\[2mm]
$(\rNK)$ Assume that $I \rho \sat {-}\phi,\Gamma \to \Delta,\phi$,  $I \rho \sat \Gamma$ and  $I \rho \sat {-} \phi$. Then 
there is an assignment $\rho'$ such that $\rho' \equiv^A_a \rho$,  $I \rho' \sat a : {-} \phi$ and $I \rho'  \sat a : \lnecb {-}\phi$.
Thus we can conclude immediately that $I \rho \not \sat \phi$. Therefore, by the hypothesis, there is $\delta \in \Delta$ such that $I \rho \sat \delta$.\\[2mm]
$(\rBP)$ Assume that $I \rho \sat \Gamma \to \Delta, t_1 \cong t_2, a: {-}(t_1 \cdot p), a: t_2 \cdot p$ and $I \rho \sat \Gamma$. 
Assume, by contradiction, that  there is  no $\psi \in \Delta \cup \{t_1 \cong t_2, a: {-}(t_1 \cdot p)\}$ such that $I \rho \sat \psi$. Then, $I \rho \sat a: t_2 \cdot p$, that is, $V(p,\rho_A(a),\rho_T(t_2))=1$. If $I \rho \sat t_1 \cong t_2$ the thesis follows immediately.  Otherwise, $\rho_T(t_1) \not \cong^I \rho_T(t_2)$. Then, by definition of $V$, $V(p,\rho_A(a),\rho_T(t_1))=0$. So
$I \rho \sat a: {-}(t_1 \cdot p)$ contradicting the assumption.\\[2mm]
$(\rKP)$ Assume that $I \rho \sat \Gamma \to \Delta, t_1 \cong t_2, {-}  (t_1 \cdot p), t_2 \cdot p$ and $I \rho \sat \Gamma$.  Assume, by contradiction that there is no 
$\psi \in \Delta \cup \{t_1 \cong t_2,{-}  (t_1 \cdot p)\}$ such that $I \rho \sat \psi$. Then, $I \rho \sat t_2 \cdot p$. We now show that $I \rho \sat {-}  (t_1 \cdot p)$ which leads to a contradiction. 
Let $\rho'$ be such that $\rho' \equiv^A_b \rho$ and  $I \rho' \sat b:t_1 \cdot p$. Then 
$V(p,\rho'_A(b),\rho_T(t_1))=1$ and so $V(p,\rho'_A(b),\rho_T(t_2))=0$ since $\rho_T(t_1) \not \cong^I \rho_T(t_2)$.
Hence, $I \rho' \sat b:{-} (t_2 \cdot p)$. Then $I \rho' \not \sat b: \lnecb {-}  (t_2 \cdot p)$ by definition of $I \rho \sat t_2 \cdot p$.
Therefore, there is $\rho'' \equiv^A_{b'} \rho'$ such that $I \rho'' \sat b \lap_p b'$ and $I \rho'' \sat b': t_2\cdot p$. Then 
$V(p,\rho''_A(b'),\rho_T(t_2))=1$ and so $V(p,\rho''_A(b'),\rho_T(t_1))=0$. Thus, $I \rho'' \sat b':{-} (t_1 \cdot p)$ and so
$I \rho' \not \sat b: \lnecb  t_1 \cdot p$. With respect to the other condition, since $I \rho \sat t_2 \cdot p$, 
let $\rho' \equiv^A_a \rho$ be such that $I \rho' \sat a:t_2 \cdot p$ and $I \rho' \sat a:\lnecb t_2 \cdot p$. Thus
$I \rho' \sat a:{-}(t_1 \cdot p)$. It remains to show that $I \rho' \sat a: \lnecb {-} (t_1\cdot p)$.
Let $\rho'' \equiv^A_{b} \rho'$ be such that $I \rho'' \sat a \lap_p b$.  
Thus,
$I \rho'' \not \sat b: {-}(t_2 \cdot p)$. 
Hence, $I \rho'' \not \sat b: t_1 \cdot p$. 
\\[2mm]
$(\rBCR)$ Assume that $I \rho' \sat b:{-} \phi, a \lap_{\var_P(\phi)} b, \Gamma \to \Delta$ for every $\rho'$ such that $\rho' \equiv^A_b \rho$, $I \rho \sat \Gamma$ and that there is no $\delta \in \Delta$ such that $I \rho \sat \delta$. 
We now show that $I \rho \sat a: \lnecb \phi$. 
Let $\sigma$ be such that 
$\sigma \equiv^A_{b} \rho$ and $I \sigma  \sat a \lap_{\var_P(\phi)} b$. Observe that $I \sigma \sat b:\phi, a \lap_{\var_P(\phi)} b, \Gamma \to \Delta$, $I \sigma \sat \Gamma$ and there is no $\delta \in \Delta$ such that $I \sigma \sat \delta$. Then $I \sigma \not \sat b: {-} \phi$. 
 \\[2mm]
$(\rBCL)$ Assume that $I \rho \sat  a: \lnecb \phi,\Gamma \to \Delta, a \lap_{\var_P(\phi)} a'$, 
$I \rho \sat  a: \lnecb \phi,\Gamma \to \Delta, a': {-} \phi$, $I \rho \sat \Gamma$ and $I \rho \sat a: \lnecb \phi$. Suppose, by contradiction, that there is no $\delta \in \Delta$ such that $I \rho \sat \delta$. Then $I \rho \sat a \lap_{\var_P(\phi)} a'$ and $I \rho \sat a':{-}\phi$. 
These facts contradict the hypothesis that $I \rho \sat a: \lnecb \phi$. \\[2mm]
$(\rBKR)$ Assume that  $I \rho' \sat b:{-} \phi, b:\lnecb {-}\phi, \Gamma \to \Delta$ for every $\rho' \equiv^A_b \rho$, $I \rho \sat \Gamma \to \Delta, \phi, a:\phi$, $I \rho \sat  \Gamma \to \Delta,\phi,a:\lnecb  \phi$ and $I \rho \sat \Gamma$. Assume, by contradiction,  that there is no $\psi \in \Delta\cup \{\phi\}$ such that $I \rho \sat \psi$. Then $I \rho \sat a:\phi$ and $I \rho \sat a:\lnecb  \phi$. We now show that $I \rho \sat \phi$ leading to a contradiction. 
Let $\sigma \equiv^A_b \rho$ be such that $I \sigma \sat b:{-} \phi$. Taking $\rho'$ as $\sigma$, by hypothesis, 
$I \sigma \not \sat b:\lnecb {-}\phi$. With respect to the second condition it is enough to take $\rho'$ as $\rho$.
\\[2mm]
$(\rKL)$ Assume that $I \rho' \sat b:\phi, b:\lnecb \phi, \Gamma \to \Delta$ for every $\rho' \equiv^A_b \rho$, $I \rho \sat \Gamma$ and $I \rho \sat \phi$. Assume, by contradiction,  that  there is no $\delta \in \Delta$ such that $I \rho \sat \delta$. Then, for every $\rho' \equiv^A_b \rho$ either $I \rho' \not \sat b:\phi$ or 
$I \rho' \not \sat b:\lnecb \phi$. Since $I \rho \sat \phi$ then 
 there is $\sigma' \equiv^A_a \rho$
such that $I \sigma' \sat a:\phi$ and $I \sigma'  \sat a:\lnecb \phi$ which contradicts the facts above. \\[2mm]
$(\rBKL)$ Assume that $I \rho \sat \phi,\Gamma \to \Delta,  a:{-} \phi$, $I \rho \sat  \phi,\Gamma \to \Delta, a:\lnecb {-}\phi$,
$ \rho \sat \Gamma$ and $I \rho \sat \phi$. Assume, by contradiction,  that  there is no $\delta \in \Delta$ such that $I \rho \sat \delta$. Thus $I \rho \sat a:{-} \phi$ and $I \rho \sat a:\lnecb {-}\phi$ contradicting the fact that $I \rho \sat \phi$.\\[2mm]
$(\CR1)$ Assume that $I \rho \sat \phi_1,\, \dots, \, \phi_n ,\Gamma \to \Delta, a::\phi$ and $I \rho \sat \Gamma$. 
If there is no $\delta \in \Delta$ such that $I \rho \sat \delta$ then $I \rho \sat a:\phi$ and so the thesis follows by definition.\\[2mm]
$(\CR2)$ Assume that $I \rho \sat\Gamma \to \Delta, \phi_i$ for $i=1,\dots,n$, $I \rho \sat a : \phi,\Gamma \to \Delta$,  $I \rho \sat \Gamma$ and $I \rho \sat  a ::(\phi_1,\, \dots, \, \phi_n /  \phi)$. Suppose, by contradiction, that  there is no $\delta \in \Delta$ such that $I \rho \sat \delta$. Hence, $I \rho \not \sat a : \phi$.  Moreover, $I \rho \sat  \phi_i$ for $i=1,\dots,n$. So, by hypothesis, $I \rho \sat a : \phi$ which is a contradiction.
\end{proof}

\begin{prop}[a.k.a.~Proposition~\ref{prop:HKsat}]
Let $H=(\Psi_L,\Psi_R)$ be a Hintikka pair over $P'$, $X'_A$ and $X'_T$. Then there are an interpretation structure $I_H$  over $P'$ and an assignment $\rho$ over $I_H$ such that $I_H \rho \sat \psi$ for each $\psi \in \Psi_L$ and $I_H \rho \not \sat \psi$
for each $\psi \in \Psi_R$. 
\end{prop}
\begin{proof} Let $H=(\Psi_L,\Psi_R)$ be a Hintikka pair. 
Consider the interpretation structure $I_{H}$ over $P'$ defined as follows:
\begin{itemize}
\item the domain $D_A$ is $X'_A$, that is, $\var_A(\Psi_L\cup \Psi_R)$,
\item the domain $D_T$ is $X'_T$, that is, $\var_T(\Psi_L\cup \Psi_R)$,
\item $(t_1,t_2)  \in \; \cong^{I_H} \quad \text{iff} \quad t_1 \cong t_2 \in  \Psi_L$,
\item $(a_1,a_2) \in \lap_p^{I_H} \quad \text{iff} \quad a_1 \lap_{p} a_2 \in  \Psi_L$,
\item $V(p,a,t) = \begin{cases}
1 & \textrm{if } a: t \cdot p \in \Psi_L
\\[1mm]
0 & \textrm{if } a: {-} (t \cdot p) \in  \Psi_L\\[1mm]
\frac 12 & \textrm{otherwise.} 
\end{cases}$

\end{itemize}
Consider also an assignment $\rho$ such that $\rho_A$  and $\rho_T$ are identities.
We show that for every $\psi\in L_{P'}^{X'_A,X'_T}$ we have 
$$
\psi \in \Psi_L   \textrm{ implies }  I_H\rho\sat \psi \quad \text{and} \quad  \psi \in \Psi_R   \textrm{ implies }  I_H\rho\not\sat \psi\,.
$$ 
\begin{enumerate}
\item $\psi$ is $t_1 \cong t_2$. 
\begin{enumerate}
\item Assume that $t_1 \cong t_2 \in \Psi_L$. Then  $(t_1,t_2) \in \; \cong^{I_H}$ and so $I_H \rho\sat t_1 \cong t_2$.
\item Assume that $t_1 \cong t_2 \in \Psi_R$. Then  $(t_1,t_2) \in \; \cong^{I_H}$ and so $I_H \rho\not \sat t_1 \cong t_2$.
\end{enumerate}
\item $\psi$ is $a_1 \lap_p a_2$. 
\begin{enumerate}
\item Assume that $a_1 \lap_p a_2 \in \Psi_L$. Then 
$(a_1,a_2) \in \lap_p^{I_H}$ and so $I_H \rho\sat a_1 \lap_p a_2$. 
\item Assume that $a_1 \lap_p a_2 \in \Psi_R$. Then 
$(a_1,a_2) \notin \lap_p^{I_H}$ and so $I_H \rho\not \sat a_1 \lap_p a_2$. 
\end{enumerate}
\item $\psi$ is $\forall x. \, x \lap_p a$.  
\begin{enumerate}
\item Assume that $\forall x. \, x \lap_p a \in \Psi_L$. Then for each $a_1\in X'_A$ we have $a_1 \lap_p a \in  \Psi_L$. 
Hence, $I_H \rho \sat a_1 \lap_p a$ for each $a_1\in D_A$ and so $I \rho \sat \forall x. \, x \lap_p a$. 
\item Assume that $\forall x. \, x \lap_p a \in \Psi_R$. Then there is $a'\in X'_A$ such that $a' \lap_p a \in  \Psi_R$. Therefore, $I \rho \not \sat a' \lap_p$ and so $I \rho \not \sat \forall x. \, x \lap_p a$.
\end{enumerate}
\item $\psi$ is $a : t \cdot p$.
\begin{enumerate}
\item Assume that $a : t \cdot p \in \Psi_L$. Then $V(p,a,t)=1$ and $I_H \rho\sat a : t \cdot p$.
\item Assume that $a : t \cdot p \in \Psi_R$. Then $V(p,a,t) \in \{0,\frac 12\}$ and so 
$I_H \rho\not \sat a : t \cdot p$.
\end{enumerate}
\item $\psi$ is $a : {-}(t \cdot p)$.
\begin{enumerate}
\item Assume that $a : {-}(t \cdot p) \in \Psi_L$. Then  $V(p,a,t)=0$. Hence $I_H \rho\sat a : {-}(t \cdot p)$. 
\item Assume that $a : {-}(t \cdot p) \in \Psi_R$.
Then $V(p,a,t) \in \{1,\frac 12\}$ and so $I_H \rho\not \sat a : {-}(t \cdot p)$.
\end{enumerate}
\item $\psi$ is $a : \lnecb \phi$.
\begin{enumerate}
\item Assume that $a : \lnecb \phi \in \Psi_L$. Then, for each $a' \in X'_A$ either $a \lap_{\var(\phi)} a' \in \Psi_R$ or $a': {-} \phi \in \Psi_R$. Hence, by the induction hypothesis, for each $a' \in X'_A$ either $I_H \rho \not \sat a \lap_{\var(\phi)} a'$ or $I_H \rho \not \sat a': {-} \phi$. Hence for each $a' \in X'_A$ if $I_H \rho \sat a \lap_{\var(\phi)} a'$ then $I_H \rho \not \sat a': {-} \phi$. Let $\rho'$ be an assignment such that $\rho' \equiv^{A'}_b \rho$ and 
$I_H \rho' \sat a \lap_{\var(\phi)} b$. Thus $I_H \rho \sat a \lap_{\var(\phi)} \rho'(b)$. Therefore, 
$I_H \rho \not \sat \rho'(b): {-} \phi$ and so $I_H \rho' \not \sat b: {-} \phi$.
\item Assume that $a : \lnecb \phi \in \Psi_R$. Then there is $a' \in X'_A$ such that $a \lap_{\var(\phi)} a', a': {-} \phi  \in \Psi_L$. So, by the induction hypothesis, there is $a' \in X'_A$ such that $I_H \rho \sat a \lap_{\var(\phi)} a'$ and
$I_H \rho \sat a': {-} \phi$. Let $\rho' \equiv^{A'}_b \rho$ be such that $\rho'(b)=a'$. Then $I_H \rho' \sat a \lap_{\var(\phi)} b$ and
$I_H \rho' \sat b: {-} \phi$. Therefore, $I \rho \not \sat a : \lnecb \phi$. 
\end{enumerate}
\item $\psi$ is $\phi \in K_P$.
\begin{enumerate}
\item Assume that $\phi \in \Psi_L$. Then there is $a' \in X'_A$ such that $a':\phi, a':\lnecb \phi \in \Psi_L$ and for every $a \in X'_A$ either $a:{-} \phi \in \Psi_R$ or $a:\lnecb{-}  \phi \in \Psi_R$. Thus, by the induction hypothesis, there is $a' \in X'_A$ such that $I_H \rho \sat a':\phi$ and $I_H \rho \sat a':\lnecb \phi$ and for each $a \in X'_A$ either $I_H \rho \not \sat  a:{-} \phi$ or $I_H \rho \not \sat a:\lnecb{-}  \phi$. We now show that $I_H \rho \sat \phi$. Let $\rho' \equiv^{A'}_b \rho$ be such that $I_H \rho' \sat b: {-} \phi$. Then $I_H \rho \sat \rho'(b): {-} \phi$. 
Hence $I_H \rho \not \sat \rho'(b):\lnecb{-}  \phi$, by hypothesis.
Therefore $I_H \rho' \not  \sat b :\lnecb{-}  \phi$. The other  condition follows straightforwardly taking into account one of the assumptions.
\item Assume that $\phi \in \Psi_R$. Then either there exists $a' \in X'_A$ such that $a':{-} \phi, a':\lnecb {-}\phi \in \Psi_L$   or 
	for each $a \in X'_A$  either $a:\phi \in \Psi_R$ or $a:\lnecb \phi  \in \Psi_R$. Hence, by the induction hypothesis,
	  either there exists $a' \in X'_A$ such that $I_H \rho \sat a':{-} \phi$ and $I_H \rho \sat a':\lnecb {-}\phi$ or 
	for each $a \in X'_A$  either $I_H \rho \not \sat a:\phi$ or $I_H \rho \not \sat a:\lnecb \phi$. We now show that $I_H \rho \not \sat \phi$. 
	Assume that 
	for each assignment $\rho'$ such that $\rho' \equiv^A_b \rho$, if  $I_H \rho' \sat b:{-} \phi$ then $I_H \rho' \not \sat b:\lnecb {-} \phi$. Then for each $a \in X'_A$  either $I_H \rho \not \sat a:\phi$ or $I_H \rho \not \sat a:\lnecb \phi$.
	Then for each $\rho'$ 
such that $\rho' \equiv^A_a \rho$ and  $I_H \rho' \sat a : \phi$ we have $I_H \rho' \not \sat a:\lnecb \phi$. 
\end{enumerate}
\item $\psi$ is $a:: (\phi_1,\dots, \phi_n / \phi)$. 
\begin{enumerate}
\item Assume that $a:: (\phi_1,\dots, \phi_n / \phi) \in \Psi_L$.~Then  either $\phi_i\in \Psi_R$  for some $i=1,\dots,n$ or 	
$a:\phi \in \Psi_L$. Then, by the induction hypothesis,  either $I_H \rho \not \sat \phi_i$  for some $i=1,\dots,n$ or 	
$I_H \rho \sat a:\phi$. Therefore, $I_H \rho  \sat a:: (\phi_1,\dots, \phi_n / \phi)$.
\item Assume that $a:: (\phi_1,\dots, \phi_n / \phi) \in \Psi_R$. Then if $\phi_1,\dots,\phi_n \in \Phi_L$ then
$a:\phi \in \Psi_R$. So, by the induction hypothesis, if $I_H \rho \sat \phi_1,\cdots,I \rho \sat \phi_n$ then
$I_H \rho \not \sat a:\phi$. Hence, $I_H \rho \not \sat a:: (\phi_1,\dots, \phi_n / \phi)$.
\end{enumerate}
\end{enumerate}
\end{proof}

\begin{prop}[a.k.a.~Proposition~\ref{prop:Hintikkafrombranch}]
Let $\Gamma'_1 \to \Delta'_1 \cdots\Gamma'_n \to \Delta'_n$ be a rooted  open analytical exhausted branch of 
an expansion $\Gamma_1 \to \Delta_1 \ \cdots$. Then $$\left(\displaystyle \bigcup_{i=1}^n \Gamma'_i,\displaystyle \bigcup_{i=1}^n \Delta'_i\right)$$ is an Hintikka pair over $\var_P(\Gamma'_n \cup \Delta'_n)$, $\var_A(\Gamma'_n \cup \Delta'_n)$ and $\var_T(\Gamma'_n \cup \Delta'_n)$.
\end{prop}
\begin{proof}
In the sequel, we denote by $\Psi_L$ and $\Psi_R$ the left and the right hand side of the pair. We now show that this pair fullfils the conditions to be a Hintikka pair:
\begin{enumerate}
\item 
Either $\beta \notin \Psi_L$ or $\beta \notin \Psi_R$ where $\beta$ is of the form $a: \phi$, $t_1 \cong t_2$ and $a_1 \lap_p a_2$. Assume by contradiction that $\beta \in \Psi_L$ and $\beta \in \Psi_R$. Hence, there are $1\leq i,j\leq n$ such that $\beta \in \Gamma'_i$ and $\beta \in \Delta'_j$. Observe that any rule preserves $\beta$. Therefore, $\beta \in \Gamma'_n \cap \Delta'_n$ and so $\Gamma'_n \to \Delta'_n$ is an axiom contradicting the hypothesis that the branch is open.
\item Either $\beta \in \Psi_L$ or $\beta \in \Psi_R$  where $\beta$ is of the form $a: \phi$, $t_1 \cong t_2$ and $a_1 \lap_p a_2$. Assume, by contradiction, that $\beta$ is not in $\Psi_L \cup \Psi_R$. Consider the branch $\Gamma'_1 \to \Delta'_1 \cdots\Gamma'_n \to \Delta'_n\ \Gamma'_n \to \Delta'_n, \beta$ using the cut rule in the last step. Observe that this branch is still analytical. Indeed $\beta$ does not occur in $(i_{R})_n$ and no fresh variables are involved. Hence we get a contradiction since the given branch is analytical and exhausted.
\item $t \cong t \notin \Psi_R$. Immediate by rule $(\ER)$ since the branch is open. ~Similarly for $\lap_p$
\item If $t_1 \cong t_2 \in \Psi_L$ then $t_2 \cong t_1 \notin \Psi_R$. Suppose that $t_1 \cong t_2 \in \Psi_L$. Then if 
$t_2 \cong t_1 \in \Psi_R$ we could conclude that the branch was not open using rule $(\ES)$. 
\item If $t_1 \cong t_2, t_2 \cong t_3 \in \Psi_L$ then $t_1 \cong t_3 \notin \Psi_R$ and similarly for $\lap_p$.
Suppose that $t_1 \cong t_2, t_2 \cong t_3 \in \Psi_L$. Then if 
$t_1 \cong t_3 \in \Psi_R$ we could conclude that the branch was not open using rule $(\ET)$. 
\item If $t_1 \cong t_2,[a:\phi]^{t_1}_{t_2} \in \Psi_L$ then $a:\phi \notin \Psi_R$. Suppose that $t_1 \cong t_2,[a:\phi]^{t_1}_{t_2} \in \Psi_L$. Then if $a:\phi \in \Psi_R$ we could conclude that the branch was not open using rule $(\EC)$. 
\item If $a_1 \lap_p a_2,a_2 \lap_p a_1,[a:\phi]^{a_1}_{a_2} \in \Psi_L$ then $a:\phi \notin \Psi_R$. The proof is similar to the one of (6).
\item If  $\forall x. \, x \lap_p a \in \Psi_R$ then there is  $a' \in X'_A$ such that $a'\lap_p a \in \Psi_R$. Assume that $\forall x.\, x \lap_p a \in \Psi_R$. Suppose, by contradiction, that 
 for every $a' \in X'_A$,  $a'\lap_p a \notin \Psi_R$. Then, rule ($\rAMR$) was not applied to $\forall x. \, x \lap_p a$ and so
 $\forall x. \, x \lap_p a \in \Delta'_n$. Consider the branch $\Gamma'_1 \to \Delta'_1 \cdots\Gamma'_n \to \Delta'_n \ \Gamma'_n \to (\Delta'_n \setminus \{\forall x. \, x \lap_p a\}),b\lap_p a$ using in the last step the rule $(\rAMR)$ where $b$ is not in $\Gamma'_n \cup \Delta'_n$. Observe that this branch is still analytical. So the original branch is not exhausted which contradicts the hypothesis.
\item If  $\forall x. \, x \lap_p a \in \Psi_L$ then $a_1 \lap_p a \in \Psi_L$ for each $a_1 \in X'_A$. Assume that $\forall x. \, x \lap_p a \in \Psi_L$. Suppose, by contradiction, that there is $a_1 \in X'_A$ such that $a_1 \lap_{p} a \notin \Psi_L$. Observe that $\forall x. \, x \lap_p a$ is in $\Gamma'_n$. Consider the branch $\Gamma'_1 \to \Delta'_1 \cdots\Gamma'_n \to \Delta'_n\ a_1 \lap_{p} a,\Gamma'_n \to \Delta'_n$ using in the last step the rule $(\rAML)$. Observe that this branch is still analytical. So the original branch is not exhausted which contradicts the hypothesis.
\item If $a: \phi \in \Psi_L$ then $a: {-} \phi \in \Psi_R$. Suppose that $a: \phi \in \Psi_L$.
Assume, by contradiction that $a: {-} \phi \notin \Psi_R$. Consider the branch $\Gamma'_1 \to \Delta'_1 \cdots\Gamma'_n \to \Delta'_n\ \Gamma'_n \to \Delta'_n, a: {-} \phi$ using in the last step the rule $(\rNB)$. Observe that this branch is still analytical. Indeed $a: {-} \phi$ does not occur in $(i_{R})_n$ and no fresh variables are involved. Hence we get a contradiction since the given branch is analytical and exhausted.
\item If $t_1 \cong t_2, a: {-} (t_1 \cdot p) \in \Psi_R$ then $a:t_2 \cdot p \in \Psi_R$. The proof is similar to one of (10).
\item If $a: \lnecb \phi \in \Psi_R$ then, there is $a' \in X'_A$ such that  $a \lap_{\var(\phi)} a', a': {-} \phi \in \Psi_L$. Assume that $a: \lnecb \phi \in \Psi_R$. Suppose, by contradiction, that 
for every $a' \in X'_A$ either $a \lap_{\var(\phi)} a' \notin \Psi_L$ or $a': {-} \phi \notin \Psi_L$. In both cases, rule $(\rBCR)$ was not applied to $a: \lnecb \phi$. Then $a: \lnecb \phi$ is in $\Delta'_n$. Consider the branch $\Gamma'_1 \to \Delta'_1 \cdots\Gamma'_n \to \Delta'_n\ a \lap_{\var(\phi)} b, b: {-} \phi,\Gamma'_n \to (\Delta'_n \setminus \{a: \lnecb \phi\})$ using in the last step the rule $(\rBCR)$ where $b$ is not in $\Gamma'_n \cup \Delta'_n$. Observe that this branch is still analytical. So the original branch is not exhausted which contradicts the hypothesis.
\item If $a: \lnecb \phi \in \Psi_L$ then,  for each $a' \in X'_A$ either $a \lap_{\var(\phi)} a' \in \Psi_R$
or $a': {-} \phi \in \Psi_R$. Assume that $a: \lnecb \phi \in \Psi_L$. Suppose, by contradiction, that there is $a' \in X'_A$ such that $a \lap_{\var(\phi)} a' \notin \Psi_R$ and  $a': {-} \phi \notin \Psi_R$. Observe that $a: \lnecb \phi$ is in $\Gamma'_n$. Consider the branch $\Gamma'_1 \to \Delta'_1 \cdots\Gamma'_n \to \Delta'_n\ \Gamma'_n \to \Delta'_n, a \lap_{\var(\phi)} a'$ using in the last step the rule $(\rBCL)$. Observe that this branch is still analytical. So the original branch is not exhausted which contradicts the hypothesis.
\item If $\phi \in \Psi_R$ then  either there exists $a' \in X'_A$ such that $a':{-} \phi, a':\lnecb {-}\phi \in \Psi_L$ or 
for each $a \in X'_A$  either $a:\phi \in \Psi_R$ or $a:\lnecb \phi  \in \Psi_R$. Suppose that $\phi \in \Psi_R$.
Assume, by contradiction, that the thesis does not hold. Then, $\phi \in \Delta'_n$ since rule $(\rBKR)$ was not applied to $\phi$ in the branch. Consider the branch $\Gamma'_1 \to \Delta'_1 \cdots\Gamma'_n \to \Delta'_n\ b:{-} \phi, b:\lnecb {-}\phi,\Gamma'_n \to \Delta'_n$ using in the last step the rule $(\rBKR)$ where $b$ is not in $\Gamma'_n \cup \Delta'_n$. Observe that this branch is still analytical. So the original branch is not exhausted which contradicts the hypothesis.
\item If $\phi \in \Psi_L$ then there is $a' \in X'_A$ such that $a':\phi, a':\lnecb \phi \in \Psi_L$ and for each $a \in X'_A$ either $a:{-} \phi \in \Psi_R$ or $a:\lnecb{-}  \phi \in \Psi_R$. Assume that $\phi \in \Psi_L$. Suppose, by contradiction and with no loss of generality, that for every $a' \in X'_A$, either $a':\phi \notin \Psi_L$ or $a':\lnecb \phi \notin \Psi_L$. 
Then rule $(\rKL)$ was not applied in the branch to $\phi$. Hence, $\phi$ is in $\Gamma'_n$.
Consider the branch $\Gamma'_1 \to \Delta'_1 \cdots\Gamma'_n \to \Delta'_n\ b:\phi, b:\lnecb \phi,\Gamma'_n \to \Delta'_n$ using in the last step the rule $(\rKL)$ where $b$ is not in $\Gamma'_n \cup \Delta'_n$. Observe that this branch is still analytical. So the original branch is not exhausted which contradicts the hypothesis.
\item If $a::(\phi_1,\dots,\phi_n / \phi) \in \Psi_R$ then if $\phi_1,\dots,\phi_n \in \Psi_L$ then
$a:\phi \in \Psi_R$. Assume that $a::(\phi_1,\dots,\phi_n / \phi) \in \Psi_R$ and $\phi_1,\dots,\phi_n \in \Psi_L$. Assume, by contradiction, that 
$a:\phi \notin \Psi_R$. Then rule $(\CR1)$ was not applied in the branch to $a::(\phi_1,\dots,\phi_n / \phi)$. Hence, $a::(\phi_1,\dots,\phi_n / \phi)$ is in $\Delta'_n$.
Consider the branch $\Gamma'_1 \to \Delta'_1 \cdots\Gamma'_n \to \Delta'_n\ \phi_1,\dots,\phi_n,\Gamma'_n \to \Delta'_n,a:\phi $ using in the last step the rule $(\CR1)$. Observe that this branch is still analytical. So the original branch is not exhausted which contradicts the hypothesis.
\item If $a::(\phi_1,\dots,\phi_n / \phi) \in \Psi_L$ then either $\phi_i\in \Psi_R$  for some $i=1,\dots,n$ or 
$a:\phi \in \Psi_L$. Suppose that $a::(\phi_1,\dots,\phi_n / \phi) \in \Psi_L$. Assume, by contradiction, that $\phi_i\notin \Psi_R$  for every $i=1,\dots,n$ and 
$a:\phi \notin \Psi_L$. Then rule $(\CR2)$ was not applied in the branch to $a::(\phi_1,\dots,\phi_n / \phi)$. Hence, $a::(\phi_1,\dots,\phi_n / \phi)$ is in $\Gamma'_n$.
Consider the branch $\Gamma'_1 \to \Delta'_1 \cdots\Gamma'_n \to \Delta'_n\ \Gamma'_n \to \Delta'_n,\phi_1$ using in the last step the rule $(\CR2)$. Observe that this branch is still analytical. So the original branch is not exhausted which contradicts the hypothesis.
\end{enumerate}
\end{proof}

\end{document}
